\pdfoutput=1
\documentclass[journal]{IEEEtran}

\usepackage{amsmath,amssymb,amsfonts,amsthm}
\usepackage{graphicx}
\usepackage[caption=false,font=footnotesize]{subfig}
\usepackage{tikz}
\usetikzlibrary{arrows.meta,positioning,fit,shapes.geometric,calc,backgrounds}
\usepackage{booktabs}
\usepackage{array}
\usepackage{cite}
\usepackage{url}
\usepackage{placeins}
\usepackage{mathtools}
\usepackage{algorithm}
\usepackage{algpseudocode}

\newtheorem{theorem}{Theorem}
\newtheorem{lemma}{Lemma}
\newtheorem{proposition}{Proposition}
\newtheorem{definition}{Definition}
\newtheorem{remark}{Remark}

\newcommand{\E}{\mathbb{E}}
\newcommand{\I}{\mathrm{I}}
\newcommand{\HH}{\mathrm{H}}
\newcommand{\Cdel}{C(d)}
\newcommand{\CN}{C_N(d)}
\newcommand{\emb}{\#\mathrm{emb}}
\newcommand{\CORD}{\mathcal{C}_N^{\mathrm{ORD}}(d)}
\newcommand{\CRLD}{\mathcal{C}_N^{\mathrm{RLD}}(d)}

\begin{document}

\title{Orbit Reduction and Learned Run Distributions\\
for Finite-Blocklength Binary Deletion Channels}

\author{Hassan Khodaiemehr and Chen Feng%
\thanks{H.~Khodaiemehr and C.~Feng are with the School of Engineering, The University of British Columbia,
Okanagan Campus, Kelowna, BC, Canada (e-mail: hassan.khodaiemehr@ubc.ca; chen.feng@ubc.ca).}%
\thanks{A preliminary conference version of parts of this work was presented at CWIT~2024
(no proceedings).}%
}

\maketitle

\begin{abstract}
In DNA data storage, racetrack memories, and packet networks, data are written as many short strands
of fixed length.
The receiver knows where each strand begins and ends, and deletions occur only inside a strand.
The relevant limit for such systems is the block capacity \(\CN\), the largest mutual information
between a length-\(N\) input and the output of a binary deletion channel with deletion probability
\(d\).
Because the boundaries are known, \(\CN/N\) is never smaller than the classical capacity \(\Cdel\).
Computing \(\CN\) is hard because the input takes \(2^N\) values.
For \(0\le d<1\), we show that the optimal input is unique, gives positive probability to every
string, and is unchanged by complementing or reversing the strings.
We then introduce the optimized run distribution (ORD), which assigns one probability to each number
of runs; it is optimal for \(N\le3\) and is solved with a numerical optimality certificate up to
\(N=16\).
For longer strands, an exact recursion for the output distribution gives unbiased rate estimates,
and an empirical Bernstein inequality turns them into confidence intervals.
Within the statistical precision, the best inputs found in this way perform no better than a
one-parameter Markov input.
For strands of \(100\)--\(200\) bits, known boundaries increase the rate by up to \(0.026\)~bits per
symbol beyond the best certified upper bound on the capacity of the unsegmented channel.
Combined with Fano's inequality, the block capacity also bounds the rate of any code that uses a
single strand; in the case tabulated in the literature, this bound is tighter than the best known
finite-length bound for strands longer than about \(90\) bits.
Conversely, block rates yield lower bounds on \(\Cdel\); at \(d=0.1\), the bound is within
\(1.1\times10^{-4}\)~bits/use of the best certified lower bound.
Finally, in our experiments InfoNCE estimates, even with the optimal critic, lose most of their
ability to distinguish good inputs from flat ones as the strand length grows.
\end{abstract}


\section{Introduction}
\label{sec:intro}

\IEEEPARstart{M}{any} storage and communication systems store or send data as a large number of short
strands of fixed length.
In DNA data storage, information is written into synthetic DNA molecules of about \(100\)--\(200\)
nucleotides; each molecule is synthesized, stored, and sequenced separately, and these steps
introduce substitutions, insertions, and deletions inside the
molecule~\cite{ErlichZielinski2017,HeckelMikutisGrass2019,ShomoronyHeckel2021,WeinbergerMerhav2022}.
Nanopore sequencing also reads one molecule at a time~\cite{MaoDiggaviKannan2018}.
In racetrack memories, data are stored in fixed-length tracks, and shift errors act as deletions
and insertions within a track~\cite{CheeRacetrack2018}.
Packet networks and file-synchronization protocols send fixed-length units whose boundaries are
marked by framing, while losses occur inside a unit~\cite{MaRamchandranTse2011,Diggavi2007}.
In all these systems, the receiver knows which output came from which strand, and a deletion
cannot move symbols from one strand to another.

The binary deletion channel (BDC) is the basic model of such errors: each transmitted bit is deleted
independently with probability \(d\), and the surviving bits arrive in order, without any indication
of where the deletions
occurred~\cite{Dobrushin1967,Mitzenmacher2009,CheraghchiCrawford2020,Mercier2010}.
Its classical capacity \(\Cdel\) assumes one arbitrarily long stream whose output boundaries are
unknown, and it is defined as the limit of \(\max\I(X^N;Y)/N\) as
\(N\to\infty\)~\cite{Dobrushin1967}.
No closed form is known, and lower and upper bounds have been improved over several
decades~\cite{MitzenmacherDrinea2006,DrineaMitzenmacher2007,FertonaniDuman2010,Papailiopoulos2026}.

For strands with known boundaries, the natural model is the \emph{block channel}, which maps one
length-\(N\) input \(X^N\) to the BDC output \(Y\) and is used independently on each strand.
It is a memoryless channel whose input alphabet is \(\{0,1\}^N\), and its capacity,
\begin{equation}
\label{eq:CN}
\CN=\max_{p(x^N)}\I(X^N;Y)\quad\text{bits per strand},
\end{equation}
is achieved by codes that span many strands, with vanishing error probability as the number of
strands grows.
The block capacity differs from \(\Cdel\) in three ways.
First, it describes a fixed strand length, not a limit as \(N\to\infty\).
Second, known boundaries are useful side information: \(\CN/N\ge\Cdel\) for every \(N\)
(Theorem~\ref{thm:sandwich}), and the difference is the rate gained by segmentation.
For example, at \(N=16\) and \(d=0.1\) a certified input achieves \(0.693\)~bits/use, whereas the
capacity of the unsegmented channel, the largest rate achievable with vanishing error probability,
is at most \(0.5725\)~bits/use~\cite{Papailiopoulos2026}.
Third, \(\CN\) is not a one-shot quantity such as the best rate of a single code of length \(N\) at a
given error probability~\cite{PolyanskiyPoorVerdu2010} or the size of a zero-error
deletion-correcting code~\cite{Levenshtein1966}; finite-length bounds for a single
strand~\cite{MorozovDuman2026FiniteLength} and codes for segmented
deletions~\cite{LiuMitzenmacher2010,LiuDuman2025Segmented,HaghighatDuman2025HalfMarker} address these
related questions.
Finite-length deletion channels are also used as tools to bound \(\Cdel\) from
above~\cite{FertonaniDuman2010,PintoRibeiro2026ParallelBA}, but there the goal is the limit
\(N\to\infty\), not a given strand length.

Problem~\eqref{eq:CN} is concave and can be solved by the Blahut--Arimoto (BA)
algorithm~\cite{Blahut1972,Arimoto1972}, but the input takes \(2^N\) values and the channel matrix has
\(2^N(2^{N+1}-1)\) entries.
Direct computation therefore stops at \(N\approx12\)--\(15\), far below practical strand lengths.
Two tools are needed to go further: input families with few parameters, and a way to evaluate their
rates when the output distribution cannot be tabulated.
For the first, we use the number of runs (maximal blocks of equal bits) of a string.
Run-length inputs have long been used as coding schemes for the unsegmented
channel~\cite{MitzenmacherDrinea2006,DrineaMitzenmacher2007,KhodaiemehrFengDuman2026RL}, but not as
input families for~\eqref{eq:CN}.
For the second, a common choice in learned communication is a neural mutual-information (MI)
estimator~\cite{Belghazi2018,Nguyen2010f,Oord2018InfoNCE,Poole2019variational,SongErmon2019,Fritschek2019,LiSheFan2023};
we show that InfoNCE estimates, even with the optimal critic, lose resolution for long strands, and
we replace them by an exact computation of the output distribution.

The main contributions are summarized in Table~\ref{tab:contrib}.
\begin{table}[t]
\caption{Summary of contributions.}
\label{tab:contrib}
\centering
\scriptsize
\setlength{\tabcolsep}{2pt}
\begin{tabular}{@{}>{\raggedright}p{0.17\columnwidth}>{\raggedright}p{0.6\columnwidth}>{\raggedright\arraybackslash}p{0.15\columnwidth}@{}}
\toprule
Topic & Result & Where \\
\midrule
Optimal input & For \(d<1\): unique, positive on every string, unchanged by complementation and reversal;
one probability per symmetry class (orbit) suffices & Thm.~\ref{thm:orbit}, Prop.~\ref{prop:unique} \\
Run-count inputs & ORD (one probability per number of runs): optimal for \(N\le3\), contains all
symmetric Markov inputs, numerically certified optima up to \(N=16\) & Prop.~\ref{prop:markov}, Sec.~\ref{sec:exact-rates} \\
Rates at any \(N\) & Exact output distribution in \(O(N^2|y|)\) operations; unbiased rate and gradient
estimates; confidence intervals that hold with probability \(\ge1-\delta\) & Prop.~\ref{prop:marginal}, Thm.~\ref{thm:conf} \\
Learning inputs & InfoNCE-type estimates cannot exceed \(\log_2B\) bits and, in our experiments,
resolve little of the gap between good and flat inputs as \(N\) grows, even with the optimal critic;
transfer and exact-gradient learning improve on the flat run-count law up to \(N=512\) and \(N=128\),
and a one-parameter Markov input performs as well within the statistical precision & Sec.~\ref{sec:est-cmp}, \ref{sec:transfer-expts} \\
Value of boundaries & At \(N=100\), rates above the best certified upper bound on \(\Cdel\) at every
tested \(d\) for which~\cite{Papailiopoulos2026} reports one; a single-strand converse tighter than~\cite{MorozovDuman2026FiniteLength} for
\(N\ge92\) &
Sec.~\ref{sec:N100}, \ref{sec:md} \\
Bounds on \(\Cdel\) & Lower bounds above Drinea--Mitzenmacher at every tested \(d\in[0.02,0.8]\), within
\(1.1\times10^{-4}\) of~\cite{Papailiopoulos2026} at \(d=0.1\), and matching the best known for small \(d\) & Sec.~\ref{sec:papa} \\
\bottomrule
\end{tabular}
\end{table}
The rest of this paper is organized as follows. The block-channel theory is developed in Section~\ref{sec:orbit}, variational and exact-gradient
learning in Sections~\ref{sec:framework} and~\ref{sec:transfer}, and the numerical results in
Section~\ref{sec:expts}; Section~\ref{sec:papa} relates the block channel to \(\Cdel\),
Section~\ref{sec:accuracy} discusses limitations, and concluding remarks are given in
Section~\ref{sec:concl}.
Relative to the earlier version~\cite{KhodaiemehrFeng2026arXiv}, the orbit reduction, whose invariance
step was incorrect, is replaced by Theorem~\ref{thm:orbit}, all ORD values are recomputed and
certified, and every nested Monte Carlo evaluation is replaced by the exact-marginal estimator with
confidence guarantees.

\section{Block Channel, Symmetries, and Run-Count Inputs}
\label{sec:orbit}

This section collects the finite-block results on which the rest of the paper relies.
We first recall the channel law and its elementary properties, then reduce~\eqref{eq:CN} by
symmetry, introduce the run-count families RLD and ORD, and derive an exact recursion for their
output marginals.
It concludes with the two tools that connect a finite-block rate to the capacity
\(\Cdel\): a block-conversion inequality and a confidence bound for simulated rates.
Table~\ref{tab:notation} summarizes the main notation.

\begin{table}[t]
\caption{Main notation.}
\label{tab:notation}
\centering
\scriptsize
\renewcommand{\arraystretch}{1.1}
\begin{tabular}{@{}>{\raggedright\arraybackslash}p{0.30\columnwidth}>{\raggedright\arraybackslash}p{0.64\columnwidth}@{}}
\toprule
Symbol & Meaning \\
\midrule
\(d\), \(N\) & Deletion probability; blocklength \\
\(X^N\), \(Y\), \(|y|\) & Input block; output (surviving bits); output length \\
\(\emb(y\hookrightarrow x)\) & Number of embeddings of \(y\) as a subsequence of \(x\) \\
\(W_x(y)\) & Channel law \(\Pr(Y=y\mid X=x)\), see~\eqref{eq:kernel} \\
\(p\), \(q\) & Input law on \(\{0,1\}^N\); induced output law \\
\(\I(p)\) & Mutual information \(\I(X^N;Y)\) in bits \\
\(\Cdel\), \(\CN\) & Capacity of the unsegmented BDC; block capacity~\eqref{eq:CN} \\
\(G\), \(\mathcal O_i\), \(K_N\) & Complement--reversal group; its orbits; number of orbits \\
\(r(x)\), \(R_\ell\) & Number of runs of \(x\); set of strings with \(\ell\) runs \\
\(w\), \(p_w\) & Run-class (or orbit) weights; the corresponding input law \\
\(\CORD\), \(\CRLD\) & Optimal ORD rate; rate of the flat run-length distribution \\
\(s_\ell(w)\) & Class-averaged divergence \(D(W_x\Vert q_w)\) over \(x\in R_\ell\) \\
\(\mathrm{pen}_N(d)\) & Length penalty \(\HH(\mathrm{Bin}(N,1-d))/N\) of Lemma~\ref{lem:fd} \\
\(f(x)\), \(Z\), \(g\) & Self-information \(-\log_2p(x)\); posterior term \(-\log_2\Pr(X\mid Y)\); upper bound on \(Z\)~\eqref{eq:Zrange} \\
\(\tau\), \(t(\tau)\) & Truncation level; exact truncation tail~\eqref{eq:tail} \\
\(n\), \(\delta\), \(\varepsilon_n\) & Sample size; confidence parameter; empirical Bernstein margin~\eqref{eq:eps} \\
\(L_\delta\) & Lower bound on \(\Cdel\) holding with probability \(\ge1-\delta\)~\eqref{eq:cap-conf} \\
\(u_\ell\), \(f_d\), \(w'\) & Normalized run index; transfer profile; transferred weights~\eqref{eq:transfer} \\
\(\Delta(N',d)\) & Rate gain of transfer over RLD at length \(N'\) \\
\(p\) (Markov input) & Flip probability of a symmetric first-order Markov input \\
\(L_{\rm P}\), \(U_{\rm P}\) & Certified lower and upper bounds on \(\Cdel\) of~\cite{Papailiopoulos2026} \\
\bottomrule
\end{tabular}
\end{table}

Each bit of \(X^N\in\{0,1\}^N\) is deleted independently with probability \(d\), and the receiver
observes the concatenation \(Y\) of the surviving bits.
The transition law \(W_x(y):=\Pr(Y=y\mid X=x)\) is the embedding-count kernel~\cite{Mitzenmacher2009}
\begin{equation}
\label{eq:kernel}
W_x(y)=\emb(y\hookrightarrow x)\,d^{N-|y|}(1-d)^{|y|},
\end{equation}
where \(\emb(y\hookrightarrow x)\) is the number of ways of obtaining \(y\) as a subsequence of
\(x\).
The count obeys the standard recursion
\(E(i,j)=E(i-1,j)+\mathbf 1\{x_i=y_j\}E(i-1,j-1)\), with \(E(i,0)=1\) and \(E(0,j)=0\) for
\(j\ge1\), where \(E(i,j)\) counts embeddings of \(y_1^j\) into \(x_1^i\); hence
\(\emb(y\hookrightarrow x)=E(N,|y|)\) is computed in \(O(N|y|)\) operations.
We refer to this recursion as the embedding recursion.
Two further facts are classical~\cite{Mitzenmacher2009}: a string of length \(m\) has at most
\(\binom Nm\) embeddings into any string of length \(N\), and
\(\sum_x\emb(y\hookrightarrow x)=2^{N-m}\binom Nm\) for every \(y\) of length \(m\).

The dense kernel has \(2^N(2^{N+1}-1)\) entries, and we use it for unrestricted inputs only for
\(N\le12\); for run-count inputs, class averages suffice (Section~\ref{sec:class-avg}).
Large-\(N\) computations use rescaling or the log domain, because embedding counts exceed 64-bit
integers.

\subsection{Symmetry reduction}
Program~\eqref{eq:CN} is the capacity, in bits per block, of the memoryless channel induced by one
length-\(N\) use.
The reduction uses two standard notions for a finite group \(G\) acting on a finite set
\(\mathcal S\), that is, a group of bijections of \(\mathcal S\) that contains the identity and is
closed under composition and inversion.

\begin{definition}[Orbit and fixed set]
\label{def:orbit}
The \emph{orbit} of \(s\in\mathcal S\) is \(G(s):=\{g(s):g\in G\}\).
Two orbits are either equal or disjoint, and therefore the orbits partition \(\mathcal S\); a function
on \(\mathcal S\) is \emph{\(G\)-invariant} if it is constant on every orbit.
The \emph{fixed set} of \(g\in G\) is \(\mathrm{Fix}(g):=\{s\in\mathcal S:g(s)=s\}\).
\end{definition}

\begin{lemma}[Burnside's lemma~\cite{Rotman1995}]
\label{lem:burnside}
The number of orbits of \(G\) on \(\mathcal S\) equals \(|G|^{-1}\sum_{g\in G}|\mathrm{Fix}(g)|\),
the average number of elements fixed by a group element.
\end{lemma}

Let \(c\) denote bitwise complementation and \(r\) bit reversal, both acting on strings of any
length.
If the index set \(\{i_1<\cdots<i_m\}\) embeds \(y\) into \(x\), then the same index set embeds
\(c(y)\) into \(c(x)\), and the reflected set \(\{N+1-i_m<\cdots<N+1-i_1\}\) embeds \(r(y)\) into
\(r(x)\); both maps are bijections between embeddings and preserve \(|y|\).
Hence, for the group \(G=\{\mathrm{id},c,r,cr\}\),
\begin{equation}
\label{eq:auto}
W_{g(x)}\bigl(g(y)\bigr)=W_x(y)\qquad\text{for all }g\in G,\;x,\;y.
\end{equation}
Mutual information is concave in the input law and is invariant under a simultaneous bijective
relabeling of both alphabets that preserves the channel; averaging an input law over a finite group
of such relabelings therefore cannot decrease it, as in the classical treatment of symmetric
channels~\cite{CoverThomas2006}.
For completeness, the argument is included in the proof of the following reduction.

\begin{theorem}[Complement--reversal orbit reduction]
\label{thm:orbit}
\label{thm:perm}
Let \(\mathcal{O}_1,\ldots,\mathcal{O}_{K_N}\) be the \(G\)-orbits of \(\{0,1\}^N\).
Then
\begin{equation}
\label{eq:red}
\CN=\max_{w\in\Delta^{K_N-1}}\I(p_w),\qquad
p_w(x)=\frac{w_i}{|\mathcal{O}_i|}\quad(x\in\mathcal{O}_i),
\end{equation}
and the map \(w\mapsto\I(p_w)\) is concave.
By Lemma~\ref{lem:burnside},
\(K_N=\tfrac14\bigl(2^N+2^{\lceil N/2\rceil}+\mathbf 1_{N\,\mathrm{even}}2^{N/2}\bigr)\).
\end{theorem}

\begin{IEEEproof}
For \(g\in G\) and an input law \(p\), let \(g_*p(x):=p(g^{-1}(x))\).
If \((X,Y)\) has joint law \(p(x)W_x(y)\), then by~\eqref{eq:auto} the pair \((g(X),g(Y))\) has joint
law \(p(g^{-1}(x))W_{g^{-1}(x)}(g^{-1}(y))=g_*p(x)W_x(y)\).
Since \(g\) is a bijection on inputs and on outputs,
\(\I(g_*p)=\I(g(X);g(Y))=\I(X;Y)=\I(p)\).
Let \(\bar p=|G|^{-1}\sum_{g\in G}g_*p\).
Because \(\I\) is concave in the input law,
\(\I(\bar p)\ge|G|^{-1}\sum_g\I(g_*p)=\I(p)\).
Moreover, \(h_*\bar p=|G|^{-1}\sum_g(hg)_*p=\bar p\) for every \(h\in G\), since \(g\mapsto hg\) permutes
\(G\); hence \(\bar p\) is constant on every orbit.
Applying this to a maximizer of~\eqref{eq:CN} shows that some maximizer is orbit-constant.
Orbit-constant laws are exactly the lifts \(p_w\) with \(w_i=p(\mathcal O_i)\), which
proves~\eqref{eq:red}; concavity in \(w\) follows because \(w\mapsto p_w\) is affine.
Finally, by Lemma~\ref{lem:burnside}, \(K_N\) is the average number of strings fixed by an element of \(G\):
the identity fixes \(2^N\) strings, \(c\) fixes none, \(r\) fixes the \(2^{\lceil N/2\rceil}\)
palindromes, and \(cr\) fixes the \(2^{N/2}\) strings with \(x_{N+1-i}=\bar x_i\) when \(N\) is even and
none when \(N\) is odd, because the middle bit would have to equal its complement.
\end{IEEEproof}

\begin{remark}
The reduction relies only on the global automorphisms~\eqref{eq:auto}.
Equality of the multisets of embedding counts of two rows does not, by itself, permit exchanging
their masses: for \(N=3\), \(d=\tfrac12\) and \(p(000,111,001)=(0.1,0.3,0.6)\), exchanging only the
masses of \(000\) and \(111\) lowers \(\I\) from \(0.7128\) to \(0.5895\)~bits.
The average of the two laws gives \(0.6930\)~bits, and therefore pairwise averaging can also decrease \(\I\).
This invalidates the permutation-based reduction of~\cite{KhodaiemehrFeng2026arXiv}.
\end{remark}

The orbit count \(K_N\) is still exponential in \(N\); consequently,~\eqref{eq:red} is a structural reduction
rather than an efficient algorithm.
Its maximizer is nevertheless well behaved, as the next result shows.

\begin{proposition}[Uniqueness, full support, and certificate]
\label{prop:unique}
Let \(0\le d<1\).
(a)~The maximizer \(p^\star\) of~\eqref{eq:CN} is unique, \(G\)-invariant, and satisfies
\(p^\star(x)>0\) for every \(x\).
(b)~For any partition \(\{\mathcal R_\ell\}\) of \(\{0,1\}^N\) and any weights \(w\) with \(w_\ell>0\) for
every \(\ell\), let
\(p_w(x)=w_\ell/|\mathcal R_\ell|\) on \(\mathcal R_\ell\), \(q_w=\sum_xp_w(x)W_x\), and
\(s_\ell(w)=|\mathcal R_\ell|^{-1}\sum_{x\in\mathcal R_\ell}D(W_x\Vert q_w)\).
Then
\begin{equation}
\label{eq:cert}
\I(p_w)=\textstyle\sum_\ell w_\ell s_\ell(w)\;\le\;\max_{v}\I(p_v)\;\le\;\max_\ell s_\ell(w).
\end{equation}
\end{proposition}

\begin{IEEEproof}
(a)~Write \(\I(p)=\HH(q)-\sum_xp(x)\HH(W_x)\) with \(q=\sum_xp(x)W_x\).
The second term is linear in \(p\).
An output of length \(N\) can arise only when no bit is deleted, and therefore \(W_{x'}(x)=(1-d)^N\mathbf 1\{x'=x\}\)
for \(|x|=N\), and therefore \(q(x)=(1-d)^Np(x)\).
Since \(d<1\), the linear map \(p\mapsto q\) is injective.
If \(p\ne p'\), then \(q\ne q'\), and strict concavity of entropy gives
\(\HH(\lambda q+(1-\lambda)q')>\lambda\HH(q)+(1-\lambda)\HH(q')\) for \(0<\lambda<1\); hence \(\I\) is
strictly concave in \(p\), and its maximizer \(p^\star\) is unique.
By the proof of Theorem~\ref{thm:orbit}, \(\I(g_*p^\star)=\I(p^\star)\), and therefore \(g_*p^\star\) is also a
maximizer and uniqueness gives \(g_*p^\star=p^\star\) for all \(g\in G\).
Suppose now that \(p^\star(x_0)=0\) for some \(x_0\), and let \(p_\varepsilon=(1-\varepsilon)p^\star+\varepsilon\mathbf 1_{x_0}\).
By the standard formula for the directional derivative of mutual
information~\cite{CoverThomas2006},
\[
\frac{d}{d\varepsilon}\I(p_\varepsilon)\Big|_{\varepsilon=0^+}=D(W_{x_0}\Vert q^\star)-\I(p^\star),
\]
where \(q^\star\) is the output law of \(p^\star\).
Since \(q^\star(x_0)=(1-d)^Np^\star(x_0)=0<(1-d)^N=W_{x_0}(x_0)\), we have
\(D(W_{x_0}\Vert q^\star)=+\infty\), and therefore \(\I(p_\varepsilon)>\I(p^\star)\) for small \(\varepsilon>0\), contradicting
optimality.
Hence \(p^\star(x)>0\) for every \(x\).

(b)~Since \(\I(p)=\sum_xp(x)D(W_x\Vert q)\) and \(p_w\) is constant on each class,
\(\I(p_w)=\sum_\ell w_\ell s_\ell(w)\).
The map \(v\mapsto\I(p_v)\) is concave, because \(v\mapsto p_v\) is affine, and its partial derivatives
at \(w\) are \(\partial\I(p_v)/\partial v_\ell|_{v=w}=s_\ell(w)-1/\ln2\), because
\(\partial\I(p)/\partial p(x)=D(W_x\Vert q)-1/\ln2\) in bits~\cite{Blahut1972,Arimoto1972}.
The first-order condition for concave functions gives, for every \(v\) in the simplex,
\[
\I(p_v)\le\I(p_w)+\sum_\ell(v_\ell-w_\ell)\bigl(s_\ell(w)-1/\ln2\bigr)=\sum_\ell v_\ell s_\ell(w),
\]
where we used \(\sum_\ell(v_\ell-w_\ell)=0\) and \(\I(p_w)=\sum_\ell w_\ell s_\ell(w)\).
The right-hand side is at most \(\max_\ell s_\ell(w)\), which proves~\eqref{eq:cert}.
(If some \(w_\ell=0\) and \(d<1\), then \(s_\ell(w)=+\infty\) and the bound is vacuous.)
\end{IEEEproof}

With singleton classes, part~(b) certifies Blahut--Arimoto; with the run classes defined next, it
certifies ORD.
Every optimum reported in this paper satisfies \((\max_\ell s_\ell-\I)/N\le10^{-6}\)~bits/use in
double precision, except Blahut--Arimoto at \(N=10\), \(d=0.7\), where the gap is \(10^{-5}\) because
some optimal masses fall below \(10^{-300}\) (weights are kept positive in the log domain).
The gaps are evaluated in floating point without directed rounding, so these optima are
\emph{numerically} certified.

\subsection{Run-count inputs: RLD and ORD}
A \emph{run} is a maximal constant substring.
Let \(r(x)\) be the number of runs of \(x\) and \(R_\ell=\{x:r(x)=\ell\}\); by the standard count of
compositions, \(|R_\ell|=2\binom{N-1}{\ell-1}\).
The run count is invariant under \(G\), and therefore each \(R_\ell\) is a union of orbits.
For \(N\ge4\) the union is in general nontrivial; for instance, \(0001\) and \(0011\) both have two
runs but lie in different orbits.
Replacing orbits by run classes therefore gives an inner approximation of~\eqref{eq:red} with only
\(N\) weights.

\begin{definition}[RLD and ORD]
\label{def:ord}
The \emph{optimized run distribution} (ORD) with weights \(w\in\Delta^{N-1}\) is
\begin{equation}
\label{eq:ord}
p_w(x)=\frac{w_{r(x)}}{|R_{r(x)}|},
\end{equation}
and \(\CORD:=\max_w \I(p_w)\).
The \emph{run-length distribution} (RLD) is the special case \(w_\ell=1/N\), with mutual
information \(\CRLD\).
\end{definition}

RLD and ORD are distributions on \(\{0,1\}^N\); they should not be confused with the i.i.d.\
run-length \emph{processes} used as infinite-blocklength coding schemes
in~\cite{MitzenmacherDrinea2006,DrineaMitzenmacher2007,KhodaiemehrFengDuman2026RL}.
Three elementary properties delimit the family.

\begin{proposition}[Properties of ORD]
\label{thm:ord-opt}
\label{prop:markov}
\label{prop:nle3}
(a)~Every input of the form \(p(x)=f(r(x))\) is an ORD law, and therefore \(\CORD\) is the largest rate among
such inputs.
(b)~A symmetric first-order Markov input with uniform first bit and flip probability \(p\) is the
ORD law with \(w_\ell=\binom{N-1}{\ell-1}p^{\ell-1}(1-p)^{N-\ell}\); in particular, the uniform input
is an ORD law.
(c)~For \(N\le3\), each run class is a single \(G\)-orbit, and therefore \(\CORD=\CN\).
\end{proposition}

\begin{IEEEproof}
(a)~If \(p(x)=f(r(x))\), set \(w_\ell=|R_\ell|f(\ell)\).
Then \(w_\ell\ge0\), \(\sum_\ell w_\ell=\sum_xp(x)=1\), and \(p(x)=w_{r(x)}/|R_{r(x)}|\); conversely, every
ORD law has this form.
(b)~A string with \(r(x)\) runs has exactly \(r(x)-1\) flips among its \(N-1\) transitions;
consequently, \(\Pr(X=x)=\tfrac12p^{r(x)-1}(1-p)^{N-r(x)}\), which depends on \(x\) only through \(r(x)\).
By~(a) the input is an ORD law, and summing over the \(|R_\ell|=2\binom{N-1}{\ell-1}\) strings of
\(R_\ell\) gives \(w_\ell\).
The choice \(p=\tfrac12\) gives the uniform input.
(c)~By Theorem~\ref{thm:orbit}, it suffices to show that each run class is a single orbit, because
ORD then coincides with the orbit family~\eqref{eq:red}.
For \(N=1\), \(R_1=\{0,1\}\) is a complement pair.
For \(N=2\), \(R_1=\{00,11\}\) and \(R_2=\{01,10\}\) are complement pairs.
For \(N=3\), \(R_1=\{000,111\}\) and \(R_3=\{010,101\}\) are complement pairs, and
\(R_2=\{001,011,100,110\}=\{001,\,c(001),\,r(001),\,cr(001)\}\) is the orbit of \(001\).
\end{IEEEproof}

\label{sec:class-avg}
For ORD laws of moderate length, the rate can be evaluated exactly without the dense kernel,
because the output marginal and the conditional entropy depend on the input only through
class averages:
\begin{equation}
\label{eq:class-avg}
\I(p_w)=\HH\Bigl(\sum_\ell w_\ell\overline\Pi_{\ell,\cdot}\Bigr)-\sum_\ell w_\ell\overline H_\ell,
\end{equation}
where \(\overline\Pi_{\ell,y}=|R_\ell|^{-1}\sum_{x\in R_\ell}W_x(y)\) and \(\overline H_\ell\) is the
average row entropy of \(R_\ell\).
Only the \(N\times(2^{N+1}-1)\) array \(\overline\Pi\) and the vector \(\overline H\) are stored.
They are accumulated by streaming all inputs through the prefix form
\(E_{xb}(y)=E_x(y)+\mathbf 1\{\mathrm{last}(y)=b\}E_x(y^-)\) of the embedding recursion, where \(y^-\)
drops the last bit of \(y\).
Writing \(t_m=d^{N-m}(1-d)^m\) and using \(\sum_{|y|=m}E_x(y)=\binom Nm\) for every \(x\), the average row
entropy is
\[
\overline H_\ell=-\sum_{m=0}^Nt_m\Bigl(L_{\ell m}+\tbinom Nm\log_2t_m\Bigr),
\]
with \(L_{\ell m}=|R_\ell|^{-1}\sum_{x\in R_\ell}\sum_{|y|=m}E_x(y)\log_2E_x(y)\),
and therefore the \(d\)-independent sums \(L_{\ell m}\) are computed once for all \(d\).
For \(N=16\) this takes a few minutes on a laptop CPU with under 1~GB of memory.
ORD optima are then computed in \(w\) by exponentiated-gradient or trust-region Newton steps with
the exact gradient \(s_\ell-1/\ln2\) and Hessian
\(-\frac1{\ln2}\sum_y\overline\Pi_{\ell y}\overline\Pi_{ky}/q_w(y)\), and certified by~\eqref{eq:cert}.

\subsection{Exact output marginals at any length}
\label{sec:exact-marginal}
Beyond \(N\approx16\) the array \(\overline\Pi\) can no longer be stored.
The output marginal of an ORD input can nevertheless be computed exactly for any \emph{given}
output, which is all that a Monte Carlo evaluation requires.

\begin{proposition}[Exact ORD marginal]
\label{prop:marginal}
Fix \(y\in\{0,1\}^m\).
For \(1\le i\le N\), \(0\le j\le m\), \(b\in\{0,1\}\), and \(1\le k\le i\), let
\(F_i(j,b,k)\) be the sum of \(2^{-i}d^{\,i-j}(1-d)^j\) over all prefixes \(x_1^i\) that end in \(b\)
and have \(k\) runs and over all deletion patterns of \(x_1^i\) whose survivors equal \(y_1^j\).
Then \(F_1(0,b,1)=d/2\), \(F_1(1,y_1,1)=(1-d)/2\), and
\begin{multline}
\label{eq:marg-dp}
F_{i+1}(j,b',k')=\tfrac12\sum_{b\in\{0,1\}}\bigl[d\,F_i(j,b,k_b)\\
+(1-d)\mathbf 1\{y_j=b'\}F_i(j-1,b,k_b)\bigr],
\end{multline}
where \(k_b=k'-\mathbf 1\{b\ne b'\}\), the second term is present only for \(j\ge1\), and
\(F_i(j,b,k)=0\) unless \(0\le j\le\min\{i,m\}\) and \(1\le k\le i\); the deletion term is kept at
\(j=0\).
With \(\pi_k=|R_k|2^{-N}\), the ORD marginal is
\begin{equation}
\label{eq:marg}
q_w(y)=\sum_{k=1}^N\frac{w_k}{\pi_k}\sum_{b}F_N(m,b,k),
\end{equation}
which costs \(O(N^2(m+1))\) operations and \(O(N(m+1))\) memory; for \(m=0\), directly
\(q_w(\varnothing)=d^N\).
\end{proposition}

\begin{IEEEproof}
The initial values follow from the definition: the one-bit prefix \(b\) has weight \(\frac12\) and
one run, and it is either deleted (weight \(d\), \(j=0\)) or kept (weight \(1-d\), \(j=1\), which requires
\(b=y_1\)).
For the recursion, consider a prefix \(x_1^{i+1}\) ending in \(b'\) with \(k'\) runs, together with a
deletion pattern whose survivors equal \(y_1^j\).
Removing the last bit gives a prefix \(x_1^i\) ending in some \(b\), with \(k_b=k'-\mathbf 1\{b\ne b'\}\)
runs, since bit \(i+1\) opens a new run exactly when \(b'\ne b\).
If bit \(i+1\) is deleted, the survivors of \(x_1^i\) equal \(y_1^j\) and the weight gains a factor
\(d\); if it is kept, it must equal \(y_j\), the survivors of \(x_1^i\) equal \(y_1^{j-1}\), and the
weight gains a factor \(1-d\).
In both cases the uniform weight gains a factor \(\frac12\).
This correspondence is a bijection between the configurations counted by \(F_{i+1}(j,b',k')\) and
those counted by the right-hand side of~\eqref{eq:marg-dp}, which proves the recursion.
At \(i=N\) and \(j=m\), summing over \(b\) gives \(\sum_bF_N(m,b,k)=2^{-N}\sum_{x\in R_k}W_x(y)\) by
the definition of the kernel~\eqref{eq:kernel}.
Since \(p_w(x)=w_k/|R_k|=(w_k/\pi_k)2^{-N}\) on \(R_k\), multiplying by \(w_k/\pi_k\) and summing over
\(k\) gives \(q_w(y)=\sum_xp_w(x)W_x(y)\).
The table \(F_i\) has \(O(N(m+1))\) entries and each is updated in \(O(1)\) time, which gives the stated
cost.
\end{IEEEproof}

Consequently, if \((X_t,Y_t)\), \(t=1,\ldots,M\), are i.i.d.\ draws from \(p_w(x)W_x(y)\), then
\begin{equation}
\label{eq:plain-mc}
\widehat\I_M=\frac1M\sum_{t=1}^M\log_2\frac{W_{X_t}(Y_t)}{q_w(Y_t)}
\end{equation}
is an unbiased estimator of \(\I(p_w)=\E[\log_2(W_X(Y)/q_w(Y))]\), with \(W_{X_t}(Y_t)\) from the
embedding recursion and \(q_w\) from~\eqref{eq:marg}.
No inner sample is involved, and therefore the estimator is free of the nested-sampling bias discussed in
Section~\ref{sec:accuracy}.
Its summands are bounded (Proposition~\ref{prop:decomp} below), and sampling \(X\) uniformly from
\(R_\ell\) instead of from \(p_w\) gives an unbiased estimate of \(s_\ell(w)\), hence of the ORD
gradient \(s_\ell-1/\ln2\).
Implemented in C with per-step rescaling, one output costs about \(2\)~ms at \(N=128\) and
\(0.13\)~s at \(N=512\) on one CPU core.
We validated~\eqref{eq:marg} against the class-averaged kernel for all outputs at \(N=8\)
(relative error below \(3\times10^{-15}\)), and~\eqref{eq:plain-mc} against exact rates at \(N=16\)
(Section~\ref{sec:transfer-expts}).

\subsection{From finite-block rates to capacity bounds}
\label{sec:conf}
The finite-block quantities above are related to \(\Cdel\) as follows.
The first statement collects known facts; the second is the block conversion of Fertonani and
Duman.

\begin{theorem}[Sandwich~\cite{Dobrushin1967,FertonaniDuman2010}]
\label{thm:sandwich}
For every \(N\ge1\) and \(d\in[0,1]\), \(\CRLD\le\CORD\le\CN\).
Moreover, \(C_{n+m}(d)\le C_n(d)+C_m(d)\), and consequently
\(\Cdel=\lim_{n\to\infty}C_n(d)/n=\inf_{n\ge 1}C_n(d)/n\), and \(\I(p)/N\le C_n(d)/n\) for every
length-\(N\) input \(p\) and every divisor \(n\) of \(N\).
\end{theorem}

The first chain holds because RLD is a feasible ORD law and ORD laws are feasible for~\eqref{eq:CN};
the remaining statements follow from subadditivity and Fekete's lemma~\cite{Dobrushin1967,Dalai2011}.
Theorem~\ref{thm:sandwich} bounds \(\Cdel\) from above by \(\CN/N\).
Its divisor bound also serves as a sanity check: neither \(\I(p)/N\) nor the mean of an unbiased
estimator of it can exceed \(C_n(d)/n\) for a divisor \(n\le12\) of \(N\), so a confidence interval
entirely above this cap indicates bias (Section~\ref{sec:accuracy}).
A finite-block rate bounds \(\Cdel\) from below only after a penalty for the unknown block
boundaries is subtracted.

\begin{lemma}[Block conversion~{\cite[Eq.~(39)]{FertonaniDuman2010}}]
\label{lem:fd}
For every length-\(N\) input law \(p\) and every \(d\in[0,1]\),
\begin{equation}
\label{eq:fd-lb}
\Cdel\ \ge\ \frac{\I(p)}{N}-\mathrm{pen}_N(d),
\end{equation}
where \(\mathrm{pen}_N(d):=\HH\bigl(\mathrm{Bin}(N,1-d)\bigr)/N\).
\end{lemma}

For fixed \(0<d<1\), the penalty behaves as \(\tfrac12\log_2(2\pi e\,Nd(1-d))/N\); it vanishes at
\(d\in\{0,1\}\), and for \(d=0.1\) it equals \(0.14\) at
\(N=16\) and \(0.003\) at \(N=2000\).
Lemma~\ref{lem:fd} therefore becomes useful only at long blocklengths, where \(\I(p)\) cannot be
computed exactly and an estimate with unknown bias is not sufficient.
The following two results show that the estimator~\eqref{eq:plain-mc} nevertheless yields bounds
that hold with a prescribed probability.

\begin{proposition}[Posterior-entropy form]
\label{prop:decomp}
Let \(p\) be a law on \(\{0,1\}^N\), let \((X,Y)\) be drawn from \(p(x)W_x(y)\), and set
\(f(x)=-\log_2p(x)\) and \(Z=-\log_2\Pr(X\mid Y)=f(X)-\log_2[W_X(Y)/q(Y)]\).
Then \(\I(p)=\HH(X^N)-\E Z\) and, almost surely,
\begin{equation}
\label{eq:Zrange}
0\ \le\ Z\ \le\ g\bigl(X,|Y|\bigr):=f(X)+\log_2\tbinom{N}{|Y|}.
\end{equation}
Moreover, \(|Y|\sim\mathrm{Bin}(N,1-d)\) is independent of \(X\).
\end{proposition}

\begin{IEEEproof}
By definition, \(\E Z=\HH(X\mid Y)\) and \(\E f(X)=\HH(X^N)\), and therefore \(\I(p)=\HH(X^N)-\HH(X\mid Y)=\HH(X^N)-\E Z\).
The lower bound \(Z\ge0\) holds because \(\Pr(X\mid Y)\le1\).
For the upper bound, fix a realization \((x,y)\) with \(p(x)W_x(y)>0\), let \(m=|y|\), and write
\(t_m=d^{N-m}(1-d)^m\).
Because \(W_x(y)>0\), the string \(y\) has at least one embedding into \(x\);
consequently, \(W_x(y)\ge t_m\) by~\eqref{eq:kernel}.
Every embedding of \(y\) into a string \(x'\) is determined by the \(m\) positions it uses;
consequently, \(\emb(y\hookrightarrow x')\le\binom Nm\) and \(W_{x'}(y)\le\binom Nm t_m\) for every \(x'\);
averaging over \(x'\sim p\) gives \(q(y)\le\binom Nm t_m\).
Therefore
\[
\Pr(x\mid y)=\frac{p(x)W_x(y)}{q(y)}\ \ge\ \frac{p(x)\,t_m}{\binom Nm t_m}=\frac{p(x)}{\binom Nm},
\]
and taking \(-\log_2\) gives \(Z\le f(x)+\log_2\binom Nm\).
Finally, the deletion pattern is an i.i.d.\ Bernoulli\((d)\) sequence drawn independently of \(X\), and
\(|Y|\) is the number of undeleted positions; hence \(|Y|\sim\mathrm{Bin}(N,1-d)\) independently of
\(X\).
\end{IEEEproof}

\begin{theorem}[Confidence bounds from simulation]
\label{thm:conf}
Let \(n\ge2\), \(0<\delta<1\), let \(Z_1,\ldots,Z_n\) be i.i.d.\ copies of \(Z\) in
Proposition~\ref{prop:decomp}, let \(\tau>0\) be fixed independently of the samples,
and let \(\bar Z_\tau\) and \(V_\tau\) denote the sample mean and the unbiased sample variance of
\(\min\{Z_t,\tau\}\).
Define
\begin{align}
\varepsilon_n(\tau,\delta)&=\sqrt{\frac{2V_\tau\ln(2/\delta)}{n}}+\frac{7\tau\ln(2/\delta)}{3(n-1)},
\label{eq:eps}\\
t(\tau)&=\E\bigl[(g(X,|Y|)-\tau)^+\bigr].
\label{eq:tail}
\end{align}
Then, with probability at least \(1-\delta\),
\begin{multline}
\label{eq:ci}
\HH(X^N)-\bar Z_\tau-\varepsilon_n(\tau,\tfrac\delta2)-t(\tau)\ \le\ \I(p)\\
\le\ \HH(X^N)-\bar Z_\tau+\varepsilon_n(\tau,\tfrac\delta2),
\end{multline}
and, with probability at least \(1-\delta\),
\begin{equation}
\label{eq:cap-conf}
\Cdel\ \ge\ L_\delta:=\frac{\HH(X^N)-\bar Z_\tau-\varepsilon_n(\tau,\delta)-t(\tau)}{N}-\mathrm{pen}_N(d).
\end{equation}
\end{theorem}

\begin{IEEEproof}
Let \(U_t=\min\{Z_t,\tau\}/\tau\), which are i.i.d.\ and take values in \([0,1]\), and let
\(\mu=\E U_1\).
The empirical Bernstein inequality of Maurer and Pontil~\cite[Thm.~4]{MaurerPontil2009} states
that, for i.i.d.\ variables in \([0,1]\) with sample mean \(\bar U\) and unbiased sample variance
\(V_U\), and for any \(\delta'\in(0,1)\),
\[
\Pr\Bigl\{\mu>\bar U+\sqrt{2V_U\ln(2/\delta')/n}+7\ln(2/\delta')/\bigl(3(n-1)\bigr)\Bigr\}\le\delta'.
\]
Since \(\bar U=\bar Z_\tau/\tau\) and \(V_U=V_\tau/\tau^2\), multiplying the deviation by \(\tau\) shows that
\(\E\min\{Z,\tau\}\le\bar Z_\tau+\varepsilon_n(\tau,\delta')\) except with probability at most
\(\delta'\).
Applying the same inequality to \(1-U_t\), which is also i.i.d.\ in \([0,1]\) with the same sample
variance, gives \(\E\min\{Z,\tau\}\ge\bar Z_\tau-\varepsilon_n(\tau,\delta')\) except with probability at
most \(\delta'\).
With \(\delta'=\delta/2\) and the union bound, both inequalities hold simultaneously with probability
at least \(1-\delta\).

Next we relate \(\E\min\{Z,\tau\}\) to \(\E Z\).
Since \(Z\ge0\), we have \(Z=\min\{Z,\tau\}+(Z-\tau)^+\).
The map \(z\mapsto(z-\tau)^+\) is nondecreasing and \(Z\le g(X,|Y|)\) by~\eqref{eq:Zrange};
consequently, \(0\le(Z-\tau)^+\le(g(X,|Y|)-\tau)^+\).
Taking expectations,
\(\E\min\{Z,\tau\}\le\E Z\le\E\min\{Z,\tau\}+t(\tau)\).
On the event of probability at least \(1-\delta\) above, this yields
\(\bar Z_\tau-\varepsilon_n(\tau,\tfrac\delta2)\le\E Z\le\bar Z_\tau+\varepsilon_n(\tau,\tfrac\delta2)+t(\tau)\),
and substituting into \(\I(p)=\HH(X^N)-\E Z\) (Proposition~\ref{prop:decomp}) gives~\eqref{eq:ci}.
For~\eqref{eq:cap-conf}, only the upper bound on \(\E\min\{Z,\tau\}\) is needed; applying it with
\(\delta'=\delta\) gives \(\I(p)\ge\HH(X^N)-\bar Z_\tau-\varepsilon_n(\tau,\delta)-t(\tau)\) with probability
at least \(1-\delta\), and Lemma~\ref{lem:fd} completes the proof.
\end{IEEEproof}

For the inputs considered in this paper, every quantity in Theorem~\ref{thm:conf} other than
\(\bar Z_\tau\) and \(V_\tau\) is explicit.
Entropies use the convention \(0\log0=0\), and \(f\), \(Z\), and \(g\) are evaluated on the support of
\(p\), which is where the bounds of Proposition~\ref{prop:decomp} hold almost surely.
For an ORD law, \(f(x)=\log_2(|R_\ell|/w_\ell)\) on \(R_\ell\) and
\(\HH(X^N)=\sum_\ell w_\ell\log_2(|R_\ell|/w_\ell)\).
For the symmetric Markov input of Proposition~\ref{prop:markov}(b),
\(f(x)=1+F\log_2\frac1p+(N-1-F)\log_2\frac1{1-p}\), where \(F\sim\mathrm{Bin}(N-1,p)\) is the number
of bit flips, and \(\HH(X^N)=1+(N-1)h_2(p)\) with \(h_2\) the binary entropy function.
Because \(|Y|\) is independent of \(X\), the law of \(g\) is a product of two explicit laws, and
\(t(\tau)\) is a finite sum.
We take \(\tau\) as the smallest support point of \(g\) with \(t(\tau)\le10^{-9}N\); in our
experiments this gives \(0.7\le\tau/N\le2.6\).
The truncation is essential, because the worst-case value of \(g\) can be several times larger and
the range enters~\eqref{eq:eps} linearly.

The guarantee refers only to the randomness of the simulation.
In every experiment the input, the blocklength \(N\), the sample size \(n\), and \(\delta\) were fixed
before sampling, and samples used to select an input were discarded.
All reported intervals and bounds use \(\delta=10^{-3}\); by the union bound, any \(K\) of them hold
simultaneously with probability at least \(1-K\delta\).
Pseudo-random numbers (PCG64) are treated as i.i.d.\ uniform, as is customary.
The recursions run in double precision with per-step rescaling, and each endpoint is widened by
\(10^{-9}\)~bits/use; spot checks against 80-bit extended precision at \(N=2000\) differed by less
than \(10^{-14}\)~bits/use.
This allowance is an empirical safeguard, not a proven error bound, and the stated coverage is
conditional on it.
The input parameters, seeds, \(\tau\), \(t(\tau)\), and unrounded endpoints of every interval are
provided with the supplementary code.
With \(n=10^5\) samples at \(N=2000\), the one-sided margin \(\varepsilon_n/N\) is a few
\(10^{-4}\)~bits/use.

For the Markov input at \(N=2000\) and \(d=0.1\) with \(n=2\times10^5\), the two terms of
\(\varepsilon_n/N\) are \(1.8\times10^{-4}\) and \(1.4\times10^{-4}\) (\(\tau/N=1.56\)), whereas Hoeffding's
inequality with the same range would give \(6.5\times10^{-3}\).

\section{Variational Estimators as Learning Objectives}
\label{sec:framework}

A common approach in learned communication is to maximize a variational MI estimate over both the
input and a critic network.
This section describes that approach for the ORD family; Section~\ref{sec:est-cmp} shows that it
fails beyond \(N\approx32\), which motivates Section~\ref{sec:transfer}.

Within the ORD family, the run-class weights are parameterized as \(w=\mathrm{softmax}(\lambda)\), and an
input is sampled by drawing \(\ell\sim w\) and then \(x\) uniformly from \(R_\ell\).
Let \(\hat\I_\theta\) denote a variational estimate with critic parameters \(\theta\).
The joint problem \(\max_{\lambda,\theta}\hat\I_\theta\) is not equivalent to \(\max_\lambda\I\):
even a population lower bound has an input-dependent gap, and a critic with constant scores makes
InfoNCE vanish for every input.
For this reason, every rate reported in this paper is computed exactly or bounded by
Theorem~\ref{thm:conf}; variational estimates are used only as optimization signals and are labeled
as such.

\label{sec:estimators}
Let \(S_{ij}=s(x_i,y_j)\) be the critic scores, in nats, on a batch of \(B\) joint pairs.
We consider
\begin{align}
\hat I_{\mathrm{NCE}}&=\frac1B\sum_{i}\ln\frac{e^{S_{ii}}}{\frac1B\sum_{j} e^{S_{ij}}},
\label{eq:nce}\\
\hat I_{\mathrm{DV}}&=\frac1B\sum_i S_{ii}-\ln\Bigl(\frac1{B^2}\sum_{i,j} e^{S_{ij}}\Bigr),
\label{eq:dv}\\
\hat I_{\mathrm{NWJ}}&=\frac1B\sum_i S_{ii}-\frac1{B^2}\sum_{i,j} e^{S_{ij}-1},
\label{eq:nwj}
\end{align}
and SMILE, which replaces \(e^{S_{ij}}\) in~\eqref{eq:dv} by \(e^{\mathrm{clip}(S_{ij},-\tau,\tau)}\) with
\(\tau=5\)~\cite{Oord2018InfoNCE,Belghazi2018,Nguyen2010f,SongErmon2019}.
By Jensen's inequality and \(\ln a\le a/e\), pointwise
\(\hat I_{\mathrm{NWJ}}\le\hat I_{\mathrm{DV}}\le\hat I_{\mathrm{NCE}}\le\ln B\); hence, for a fixed critic on
fresh pairs, all three are lower bounds in expectation, although realizations can exceed the mutual
information, and InfoNCE is capped by \(\log_2B\)~bits~\cite{Poole2019variational}.
Clipped SMILE is not a lower bound even in population~\cite{SongErmon2019}.
Each complete objective is divided by \(\ln2\), and the product terms
in~\eqref{eq:dv}--\eqref{eq:nwj} average over all \(B^2\) pairs, a fraction \(1/B\) of which are joint.

The critic is a two-tower network: a two-layer GELU network maps \(X\) and its run features to
\(\mathbb{R}^{128}\), a bidirectional GRU (hidden size \(64\)) embeds the padded \(Y\) and its length, and a
three-layer network scores the two embeddings, their difference, and their product.

For \(N\le8\), the exact kernel is available, and we use a hybrid procedure.
The critic is trained for 50 epochs on batches drawn from the current \(w\), and every fifth epoch
\(\lambda\) takes five Adam steps along the central-difference gradient (logit step \(0.05\)) of the
exactly evaluated objective \(\I(p_{w(\lambda)})/N+0.02\,\HH(w)\); this exact-objective ascent supplies inputs of known rate on which
the estimators can be calibrated.

For larger \(N\), \(\lambda\) and \(\theta\) are trained jointly by AdamW, and the \(\lambda\)-gradient is the
batch score-function estimate
\begin{equation}
\label{eq:rf}
\hat g=\bigl(\hat I_B-b\bigr)\,\frac1B\sum_{k=1}^B\nabla_\lambda\log w_{\ell_k},
\end{equation}
augmented by a small entropy bonus, where \(\hat I_B\) is the InfoNCE value of the batch and \(b\) is an
exponential moving average of earlier batch values.
For a fixed critic, and since \(b\) does not depend on the current batch, \(B\hat g\) is an unbiased
estimate of \(\nabla_\lambda\E[\hat I_B]\); it ignores the adaptation of \(\theta\) and targets the
surrogate rather than \(\I\), and the factor \(1/B\) reduces the weight of this term relative to the
entropy bonus.

\section{Transfer and Exact-Gradient Learning at Large \(N\)}
\label{sec:transfer}

For \(N\le16\), certified ORD optima are computable, whereas beyond that length neither BA nor
class-averaged kernels are practical, and the score-function procedure of
Section~\ref{sec:framework} fails (Section~\ref{sec:scale}).
We therefore consider two alternatives that use the exact marginal of
Proposition~\ref{prop:marginal}: transferring the shape of small-\(N\) optima to a target length,
and learning the weights directly from unbiased estimates of the exact gradient.

In profile transfer, the \(\ell\)-th weight of a length-\(N\) ORD input is placed at the normalized run
index \(u_\ell=(\ell-\tfrac12)/N\in(0,1)\).
For each \(d\in\{0.1,0.2,\ldots,0.9\}\) and \(N\in\{4,6,8,10\}\), near-optimal weights \(w^{(N,d)}\)
were obtained by 80 Adam steps along the central-difference gradient of the exactly evaluated rate
with entropy weight \(0.005\); for \(d\in\{0.1,0.3,0.5\}\), each is within \(3\times10^{-4}\)~bits/use
of the certified optimum.
The profile \(f_d\) is obtained by linearly interpolating the masses \(w^{(N,d)}_\ell\) at \(u_\ell\) onto
200 equispaced points of \([10^{-3},1]\), extending them as constants beyond the extreme midpoints,
averaging the four interpolants, flooring at \(10^{-12}\), and normalizing to unit integral by the
trapezoidal rule; \(f_d\) is the linear interpolant of these values, extended as a constant on
\([0,10^{-3}]\).
Given \(f_d\) and a target length \(N'\), the transferred input is
\begin{equation}
\label{eq:transfer}
w'_\ell \propto \max\bigl\{f_d\bigl((\ell-\tfrac12)/N'\bigr),10^{-12}\bigr\},\qquad
\textstyle\sum_{\ell=1}^{N'} w'_\ell=1,
\end{equation}
which is a valid ORD input at length \(N'\).
As \(d\) increases, \(f_d\) moves toward smaller \(u\), that is, toward fewer and longer runs.
For a fixed profile the peak mass decays as \(\Theta(1/N')\) (from about \(0.12\) at \(N'=16\) to
\(0.004\) at \(N'=512\)), and therefore the effective number of classes \(1/\sum_\ell w_\ell'^2\) grows linearly
in \(N'\).

Exact-gradient learning, the second alternative, relies on the fact that sampling \(X\) uniformly
from \(R_\ell\) in~\eqref{eq:plain-mc} gives unbiased estimates of \(s_\ell(w)\), and hence of the exact
ORD gradient \(s_\ell-1/\ln2\).
We use them in the stochastic mirror-ascent update
\begin{equation}
\label{eq:smd}
w_\ell\leftarrow w_\ell\,2^{\eta\tilde s_\ell}\Big/\sum_k w_k2^{\eta\tilde s_k},\qquad \eta=0.5,
\end{equation}
which reduces to the BA update when \(\tilde s=s\) and \(\eta=1\).
Here \(\tilde s\) is a degree-8 Chebyshev fit in \(u_\ell\) to per-class estimates from \(n\) samples per
class (\(n=40\) for \(N\le64\) and \(n=20\) for \(N=128\)); the fit exploits the smoothness of
\(s_\ell\) in \(\ell\) to reduce variance.
The fitted \(\tilde s\) is in general biased, so~\eqref{eq:smd} is a heuristic without a convergence
guarantee; the reported rates are evaluated independently.
Starting from the transferred input, we run 25 iterations for \(N\le64\) and 14 for \(N=128\).
At \(N=16\), the procedure returns an input within \(3\times10^{-4}\)~bits/use of the certified
optimum.
Final rates are evaluated on fresh samples, independent of those used for learning.

\section{Numerical Results for Finite-Block Rates}
\label{sec:expts}

This section evaluates the families and procedures described above.
Throughout, rates are either exact (up to floating-point error) or reported as \(99.9\%\) confidence
intervals from Theorem~\ref{thm:conf} (\(\delta=10^{-3}\)); lower ends are rounded down and upper ends up.

\subsection{Certified rates for \(N\le16\)}
\label{sec:exact-rates}
Table~\ref{tab:exact} reports certified ORD and BA rates.
Through \(N=10\), the largest relative gap between BA and ORD is \(1.96\%\) (\(d=0.8\)), and it is
below \(0.3\%\) for \(d\le0.5\).
ORD improves substantially on RLD; at \(N=16\), for instance, the certified ORD rates are \(0.6926\),
\(0.3651\), and \(0.2145\)~bits/use at \(d=0.1,0.3,0.5\), against \(0.5900\), \(0.3058\), and \(0.1562\)
for RLD.
The inputs returned by the hybrid procedure of Section~\ref{sec:framework} are within
\(1.7\times10^{-3}\)~bits/use of the ORD optimum.

Proposition~\ref{prop:markov}(b) identifies the symmetric Markov inputs as a one-parameter
subfamily of ORD.
Table~\ref{tab:markov} shows that the best such input is within \(0.0013\)~bits/use of the ORD optimum
at every tested \((N,d)\), and within \(2.2\%\) of \(\CN\) at \(N=10\).
To this accuracy, the remaining \(N-2\) degrees of freedom of ORD are not used by the optimum: a
single flip probability \(p^\star(N,d)\) captures almost all of the finite-block capacity, and ORD
serves to certify this.

\begin{table}[t]
\caption{Best symmetric Markov input (one parameter, exact rate) versus numerically certified ORD
and BA optima (bits/use).}
\label{tab:markov}
\centering
\scriptsize
\begin{tabular}{@{}cccccc@{}}
\toprule
$N$ & $d$ & $p^\star$ & Markov$(p^\star)$ & ORD & BA \\
\midrule
10 & 0.1 & 0.438 & 0.7295 & 0.7297 & 0.7301 \\
10 & 0.2 & 0.366 & 0.5404 & 0.5411 & 0.5421 \\
10 & 0.3 & 0.290 & 0.4086 & 0.4097 & 0.4108 \\
10 & 0.4 & 0.222 & 0.3153 & 0.3163 & 0.3170 \\
10 & 0.5 & 0.165 & 0.2465 & 0.2472 & 0.2476 \\
10 & 0.6 & 0.119 & 0.1926 & 0.1930 & 0.1942 \\
10 & 0.7 & 0.080 & 0.1476 & 0.1479 & 0.1500 \\
10 & 0.8 & 0.043 & 0.1073 & 0.1075 & 0.1096 \\
10 & 0.9 & 0.009 & 0.0666 & 0.0666 & 0.0670 \\
\midrule
16 & 0.1 & 0.435 & 0.6924 & 0.6926 & -- \\
16 & 0.3 & 0.280 & 0.3638 & 0.3651 & -- \\
16 & 0.5 & 0.156 & 0.2133 & 0.2145 & -- \\
\bottomrule
\end{tabular}

\end{table}

\begin{table}[t]
\caption{Numerically certified ORD and BA rates, RLD, and the exact rate of the hybrid learner's
output (bits/use).}
\label{tab:exact}
\centering
\scriptsize
\setlength{\tabcolsep}{3pt}
\begin{tabular}{@{}ccccccc@{}}
\toprule
$N$ & $d$ & ORD & RLD & BA & learned & ORD$-$learned \\
\midrule
4 & 0.1 & 0.8054 & 0.7698 & 0.8055 & 0.8052 & 0.0003 \\
4 & 0.3 & 0.5183 & 0.4843 & 0.5184 & 0.5181 & 0.0002 \\
4 & 0.5 & 0.3323 & 0.2776 & 0.3323 & 0.3315 & 0.0008 \\
6 & 0.1 & 0.7724 & 0.7159 & 0.7725 & 0.7716 & 0.0007 \\
6 & 0.3 & 0.4674 & 0.4250 & 0.4677 & 0.4667 & 0.0007 \\
6 & 0.5 & 0.2908 & 0.2328 & 0.2909 & 0.2893 & 0.0015 \\
8 & 0.1 & 0.7483 & 0.6774 & 0.7485 & 0.7479 & 0.0004 \\
8 & 0.3 & 0.4339 & 0.3861 & 0.4345 & 0.4328 & 0.0011 \\
8 & 0.5 & 0.2651 & 0.2062 & 0.2654 & 0.2634 & 0.0017 \\
\midrule
10 & 0.1 & 0.7297 & 0.6481 & 0.7301 & -- & -- \\
10 & 0.3 & 0.4097 & 0.3581 & 0.4108 & -- & -- \\
10 & 0.5 & 0.2472 & 0.1882 & 0.2476 & -- & -- \\
10 & 0.7 & 0.1479 & 0.0848 & 0.1500 & -- & -- \\
10 & 0.8 & 0.1075 & 0.0488 & 0.1096 & -- & -- \\
10 & 0.9 & 0.0666 & 0.0208 & 0.0670 & -- & -- \\
\midrule
15 & 0.1 & 0.6974 & 0.5976 & -- & -- & -- \\
15 & 0.3 & 0.3707 & 0.3124 & -- & -- & -- \\
15 & 0.5 & 0.2186 & 0.1602 & -- & -- & -- \\
16 & 0.1 & 0.6926 & 0.5900 & -- & -- & -- \\
16 & 0.3 & 0.3651 & 0.3058 & -- & -- & -- \\
16 & 0.5 & 0.2145 & 0.1562 & -- & -- & -- \\
\bottomrule
\end{tabular}
\end{table}

\subsection{Variational estimators}
\label{sec:est-cmp}
To separate estimation from input optimization, we fix the hybrid output at \(N=8\), \(d=0.3\), whose
exact rate is \(0.4328\)~bits/use, train each critic for 40 epochs with \(B=1024\), and evaluate on fresh
batches.
All four estimates lie below the exact rate; InfoNCE, DV, and NWJ are lower bounds in expectation,
whereas SMILE is not.
The estimates (mean \(\pm\) standard deviation over evaluation batches) are \(0.317\pm0.007\) for
InfoNCE, \(0.310\pm0.007\) for DV, \(0.307\pm0.008\) for SMILE, and \(0.112\pm0.002\) for NWJ.
InfoNCE therefore has the smallest error (\(-0.116\)~bits/use), DV and SMILE are close to it, and NWJ
is the most conservative (\(-0.321\)); InfoNCE is used at large \(N\).

\label{sec:scale}
At larger blocklengths, InfoNCE serves as the learning objective.
Training \(\lambda\) and the InfoNCE critic jointly with~\eqref{eq:rf} for \(N\in\{32,64,128\}\)
(\(B=2048,2048,1024\)) leaves the weights at RLD (\(\max_\ell|Nw_\ell-1|\le0.012\)), whose rates in
Table~\ref{tab:transfer-exactq} the InfoNCE values underestimate by a factor of \(2.5\)--\(5\) at
\(N=32\) and by more than \(100\) at \(N=128\) (estimates of at most \(0.0011\)~bits/use).

This failure worsens as \(N\) grows.
For any critic, the InfoNCE value of a batch of size \(B\) is at most \(\log_2B\)~bits, i.e.,
\(\log_2B/N\)~bits/use~\cite{Poole2019variational}, whereas \(\I(X^N;Y)\) grows linearly in \(N\);
avoiding this ceiling requires \(B\ge2^{\I(X^N;Y)}\) (a necessary, not sufficient, condition).
Similarly, no distribution-free high-confidence lower bound computed from \(n\) samples alone exceeds
\(O(\log n)\)~\cite{McAllester2020formal}.
To separate this limit from critic training, we evaluated the output-anchored form of~\eqref{eq:nce},
\(\frac1B\sum_i\log_2\bigl[W_{x_i}(y_i)/\bigl(\frac1B\sum_jW_{x_j}(y_i)\bigr)\bigr]\), i.e., the
oracle critic \(\ln W_x(y)\) with the roles of inputs and outputs exchanged; this critic is optimal for
that form, because the posterior probability that \(x_j\) produced \(y_i\) is proportional to
\(W_{x_j}(y_i)\)~\cite{Poole2019variational}, and the \(\log_2B\) ceiling still applies.
For each \((N,d)\) we used batch sizes that grow with the blocklength, \(B=16N\) and \(B=64N\), a
fixed Markov input (flip probability \(0.438\), \(0.366\), \(0.247\), \(0.189\), \(0.140\), \(0.101\),
\(0.080\), \(0.049\), and \(0.010\) for \(d=0.1,\ldots,0.9\), selected at \(N=100\)), and between \(48\) and
\(8192\) independent batch rows, and we compared the result with the rate of the same input
estimated by~\eqref{eq:plain-mc}.
Fig.~\ref{fig:nce}(a) shows that, for every \(d\), the fraction of the rate recovered by the oracle
InfoNCE is close to one at short blocklengths and then decays approximately as \(\log_2B/\I(X^N;Y)\):
with \(B=64N\) it is \(0.17\) at \(N=128\) and \(0.05\) at \(N=512\) for \(d=0.1\), \(0.25\) at
\(N=512\) for \(d=0.5\), and with \(B=16N\) it drops to \(0.52\) at \(N=2048\) even for \(d=0.9\),
where the rate is only \(0.014\)~bits/use.
Fig.~\ref{fig:nce}(b) shows that all tested \((N,d,B)\) lie close to the envelope
\(\min\{1,\log_2B/\I(X^N;Y)\}\): the decay sets in, for every \(d\), once the block information exceeds
\(\log_2B\), which happens at larger \(N\) for larger \(d\) because the rate per symbol is smaller;
the envelope is an upper bound, not an exact law.

Input learning, however, needs the objective to rank inputs.
Fig.~\ref{fig:nce}(c) reports the fraction of the true gap between the best Markov input and RLD
that is visible to the oracle InfoNCE.
It is close to one at short blocklengths and then falls: at \(N=256\) it is at most \(0.04\) for
\(d\le0.3\), between \(0.11\) and \(0.28\) for \(d\in\{0.4,0.5\}\), and below \(0.5\) for \(d=0.6\), and at
\(N=512\) it is \(0.34\) and \(0.57\) for \(d=0.7\) and \(0.8\).
Empirically, once both tested inputs reach the ceiling, the oracle objective resolves only a small
fraction of the gap between them; the ceiling itself bounds each value, not their difference.
For \(d=0.9\), the block information of RLD is only about \(1\)~bit at \(N=256\), far below
\(\log_2B\), and the gap is still resolved; by Fig.~\ref{fig:nce}(b), we expect it to be lost once
this information exceeds \(\log_2B\), which for \(d=0.9\) happens only at blocklengths of several
thousand.
This is consistent with the learned weights above, which remain at RLD for
\(d\in\{0.1,0.3,0.5\}\).
The exact marginal of Proposition~\ref{prop:marginal} avoids this limit because it evaluates
\(\log P(y)\) directly rather than from a batch; the methods of Section~\ref{sec:transfer} use it.

\begin{figure*}[!t]
\centering
\includegraphics[width=\textwidth]{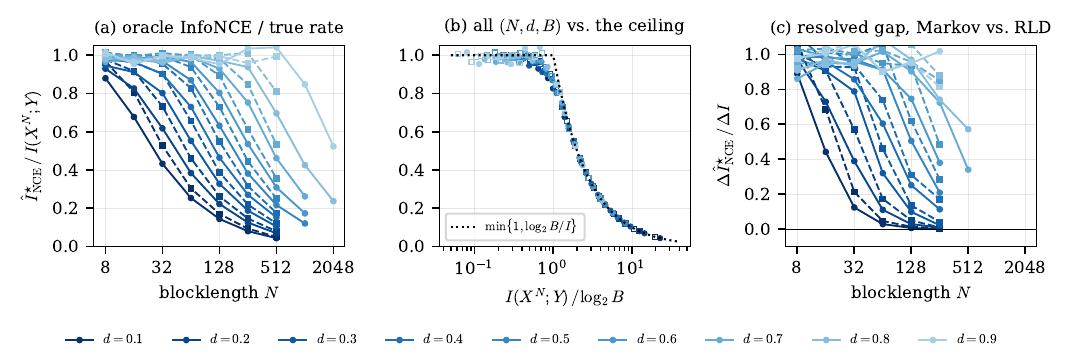}
\caption{Output-anchored InfoNCE with the oracle critic $\ln W_x(y)$ and batch size growing with the
blocklength ($B=16N$, circles and solid lines; $B=64N$, squares and dashed lines), for a fixed
Markov input and all nine values of $d$ (darker: smaller $d$).
(a)~Fraction of the rate $\I(X^N;Y)$ recovered.
(b)~The same points against $\I(X^N;Y)/\log_2B$, with the envelope $\min\{1,\log_2B/\I\}$ implied by
the InfoNCE ceiling~\cite{Poole2019variational}.
(c)~Fraction of the true gap $\Delta I=\I(p_{\rm Markov})-\I(p_{\rm RLD})$ resolved by the oracle
InfoNCE ($N\le256$; $N=512$ for $d\in\{0.7,0.8\}$).
True rates are the unbiased estimates~\eqref{eq:plain-mc} from $3000$ to $20{,}000$ samples.}
\label{fig:nce}
\end{figure*}

\subsection{Transfer and exact-gradient learning}
\label{sec:transfer-expts}
At \(N'=16\), where the class-averaged kernel is still available, transfer improves on RLD by
\(\Delta:=\I(p_{w'})/N'-\I(p_{\mathrm{RLD}})/N'=0.078\), \(0.049\), and \(0.053\)~bits/use at
\(d=0.1,0.3,0.5\), and it is within \(0.025\), \(0.010\), and \(0.006\) of the certified optimum.
These values also validate Theorem~\ref{thm:conf}: the \(99.9\%\) intervals computed from
\(n=20{,}000\) simulated samples contain all six exact rates.

For \(32\le N'\le512\), Table~\ref{tab:transfer-exactq} and Fig.~\ref{fig:delta} report
intervals from \(n=20{,}000\) samples (\(N'\le64\)), \(10{,}000\) (\(N'=128\)), \(6000\) (\(N'=256\)), and
\(4000\) (\(N'=512\)).
Subtracting unrounded endpoints gives an interval for \(\Delta\); since each endpoint fails with
probability at most \(\delta/2\), each one-sided bound on \(\Delta\) holds with probability
\(99.9\%\) and the two-sided interval with \(99.8\%\).
The point estimates of \(\Delta\) lie in \([0.085,0.093]\), \([0.045,0.057]\), and \([0.032,0.045]\)~bits/use
at \(d=0.1\), \(0.3\), and \(0.5\).
The lower bounds confirm \(\Delta>0\) at every tested length except \(N'=512\), \(d=0.5\), where \(4000\)
samples leave an interval containing zero; at \(d=0.1\) they give \(\Delta\ge0.063\) up to
\(N'=512\).
In particular, the plateau at \(d=0.1\) reported in~\cite{KhodaiemehrFeng2026arXiv} does not occur
(Section~\ref{sec:accuracy}).

Learning with~\eqref{eq:smd} improves further on transfer (Table~\ref{tab:learned}).
At \(d=0.1\), the point gains are \(0.045\), \(0.060\), and \(0.069\)~bits/use for \(N=32,64,128\), and the
intervals guarantee gains of at least \(0.031\), \(0.049\), and \(0.057\).
At \(d\in\{0.3,0.5\}\), the point gains are \(0.011\)--\(0.032\)~bits/use and are confirmed by the
intervals for \(N\ge64\) at \(d=0.3\) and for \(N=128\) at \(d=0.5\).
The gain grows with \(N\) because the learned weights, unlike the transferred ones, concentrate
: as a distribution of the run fraction \(u=(\ell-\tfrac12)/N\), they are centred
at \(u^\star(d)\approx0.43\), \(0.27\), \(0.145\), nearly independently of \(N\), with a standard deviation
between \(0.50/\sqrt N\) and \(0.61/\sqrt N\) for \(N\in\{32,64,128\}\), and the effective number of classes
\(1/\sum_\ell w_\ell^2\) grows from \(8\)--\(12\) at \(N=32\) to \(20\)--\(24\) at \(N=128\).
The effective number of active classes therefore grows like \(\sqrt N\), whereas a fixed profile
spreads its mass over \(\Theta(N)\) classes.
This is the same order of concentration as for a Markov input, whose run count is
\(1+\mathrm{Bin}(N-1,p)\), so that its run fraction has standard deviation
\(\sqrt{(N-1)p(1-p)}/N\approx(0.35\)--\(0.50)/\sqrt N\) for the relevant \(p\).
Accordingly, the intervals of the learned input and of the best tested Markov input on a five-point
grid of flip probabilities overlap at every \((N,d)\) in Table~\ref{tab:learned} (e.g., their
difference lies in \([-0.009,0.013]\) at \(N=128\), \(d=0.5\)).
Overlap does not establish equality, and no learned input bounds \(\CORD\) from above; with
Table~\ref{tab:markov}, these results show only that, for \(N\le128\), no tested run-count input
improves measurably on a first-order Markov input.

\begin{figure}[t]
\centering
\includegraphics[width=0.75\columnwidth]{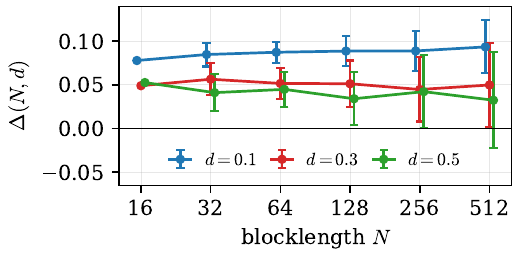}
\caption{Gain of transfer over RLD, $\Delta(N,d)$ (bits/use): exact at $N=16$, and $99.8\%$ confidence
intervals (each end at $99.9\%$; Theorem~\ref{thm:conf}) for $N\ge32$. The point estimates stay
positive and roughly constant in $N$; every interval excludes zero except at $N=512$, $d=0.5$.}
\label{fig:delta}
\end{figure}

\begin{table}[t]
\caption{Transferred ORD versus RLD (bits/use).
Rows $N=16$ are exact; other rows are $99.9\%$ confidence intervals from Theorem~\ref{thm:conf}
(lower ends rounded down, upper ends rounded up); the interval for $\Delta$ is obtained by
subtracting unrounded endpoints and has confidence $99.8\%$ ($99.9\%$ for each end).}
\label{tab:transfer-exactq}
\centering
{\scriptsize
\setlength{\tabcolsep}{2.5pt}
\begin{tabular}{@{}ccccc@{}}
\toprule
$N$ & $d$ & transfer & RLD & $\Delta$ \\
\midrule
16 & 0.1 & 0.668 & 0.590 & $+0.078$ \\
32 & 0.1 & [0.597, 0.611] & [0.512, 0.527] & $[+0.071, +0.099]$ \\
64 & 0.1 & [0.549, 0.562] & [0.462, 0.475] & $[+0.075, +0.100]$ \\
128 & 0.1 & [0.516, 0.533] & [0.427, 0.446] & $[+0.071, +0.106]$ \\
256 & 0.1 & [0.493, 0.516] & [0.402, 0.428] & $[+0.065, +0.113]$ \\
512 & 0.1 & [0.481, 0.510] & [0.386, 0.418] & $[+0.063, +0.124]$ \\
\midrule
16 & 0.3 & 0.355 & 0.306 & $+0.049$ \\
32 & 0.3 & [0.290, 0.309] & [0.233, 0.253] & $[+0.038, +0.075]$ \\
64 & 0.3 & [0.245, 0.264] & [0.193, 0.212] & $[+0.033, +0.070]$ \\
128 & 0.3 & [0.216, 0.243] & [0.164, 0.193] & $[+0.024, +0.078]$ \\
256 & 0.3 & [0.187, 0.224] & [0.141, 0.180] & $[+0.008, +0.082]$ \\
512 & 0.3 & [0.172, 0.220] & [0.121, 0.171] & $[+0.001, +0.098]$ \\
\midrule
16 & 0.5 & 0.209 & 0.156 & $+0.053$ \\
32 & 0.5 & [0.153, 0.176] & [0.113, 0.134] & $[+0.020, +0.063]$ \\
64 & 0.5 & [0.126, 0.148] & [0.082, 0.103] & $[+0.024, +0.066]$ \\
128 & 0.5 & [0.092, 0.124] & [0.058, 0.089] & $[+0.003, +0.065]$ \\
256 & 0.5 & [0.082, 0.125] & [0.040, 0.083] & $[+0.000, +0.084]$ \\
512 & 0.5 & [0.062, 0.118] & [0.030, 0.085] & $[-0.023, +0.088]$ \\
\bottomrule
\end{tabular}
}
\end{table}

\begin{table}[t]
\caption{Large-$N$ rates (bits/use): $99.9\%$ confidence intervals (Theorem~\ref{thm:conf},
$n=20{,}000$ fresh samples) for ORD weights learned by~\eqref{eq:smd} and for the best tested
symmetric Markov input; point estimates for transfer and RLD, whose intervals are in
Table~\ref{tab:transfer-exactq}.}
\label{tab:learned}
\centering
\scriptsize
\begin{tabular}{@{}cccccc@{}}
\toprule
$N$ & $d$ & learned ORD & Markov$(p^\star)$ & transfer & RLD \\
\midrule
32 & 0.1 & [0.642, 0.656] & [0.640, 0.654] & 0.604 & 0.519 \\
32 & 0.3 & [0.307, 0.322] & [0.308, 0.322] & 0.299 & 0.243 \\
32 & 0.5 & [0.167, 0.184] & [0.166, 0.182] & 0.165 & 0.123 \\
\midrule
64 & 0.1 & [0.610, 0.621] & [0.611, 0.622] & 0.556 & 0.468 \\
64 & 0.3 & [0.276, 0.289] & [0.276, 0.288] & 0.255 & 0.203 \\
64 & 0.5 & [0.145, 0.159] & [0.143, 0.155] & 0.137 & 0.092 \\
\midrule
128 & 0.1 & [0.590, 0.599] & [0.590, 0.599] & 0.525 & 0.436 \\
128 & 0.3 & [0.256, 0.267] & [0.258, 0.267] & 0.230 & 0.178 \\
128 & 0.5 & [0.129, 0.141] & [0.128, 0.138] & 0.108 & 0.074 \\
\bottomrule
\end{tabular}

\end{table}

\subsection{Strands of length \(100\)--\(200\) and the value of known boundaries}
\label{sec:N100}
Strand lengths between \(100\) and \(200\) are representative of synthesized DNA
oligonucleotides~\cite{ErlichZielinski2017,HeckelMikutisGrass2019} and of short framed packets.
Fig.~\ref{fig:N100-grid}(a) and Table~\ref{tab:N100-conv} report intervals at \(N=100\) for
\(d\in\{0.01,0.02,0.03,0.05,0.1,\ldots,0.9\}\) and three inputs: the best tested Markov input, whose flip
probability was selected in a separate run, transfer, and RLD; transfer is reported only for
\(d\ge0.1\), the range covered by its profiles.
The point estimates satisfy Markov \(>\) transfer \(>\) RLD at every \(d\); the intervals confirm the
first inequality for \(d\le0.5\) and the second for \(d\le0.8\), and all intervals lie below the
cap \(C_{10}(d)/10\ge C_{100}(d)/100\) of Theorem~\ref{thm:sandwich}, given by the dual value
\(\max_x D(W_x\Vert q)\) of~\eqref{eq:cert} at the final BA iterate and rounded upward.
At small \(d\) the flat run-count law is far from optimal (\(0.705\) versus \(0.940\)~bits/use at
\(d=0.01\)), whereas the best Markov input is nearly uniform (\(p^\star\in[0.46,0.50]\)).
Since the Markov rate is achievable on the block channel, the lower end of its interval is a lower
bound on \(C_{100}(d)/100\) that holds with probability at least \(1-10^{-3}\).
Table~\ref{tab:N100-conv} also gives Markov intervals at \(N=150\) and \(N=200\), obtained with
\(n=10^5\) fresh samples after the flip probability had been selected on a nine-point grid in a
separate pilot run.

These lower bounds can be compared directly with the unsegmented channel.
The rates achievable with vanishing error probability on the unsegmented BDC are limited by
\(\Cdel\le U_{\rm P}\), the certified upper bound of~\cite{Papailiopoulos2026}.
As shown in Fig.~\ref{fig:N100-grid}(b) and in the last three columns of Table~\ref{tab:N100-conv},
at \(N=100\) the lower end of the Markov interval exceeds \(U_{\rm P}\) at every \(d\) for
which~\cite{Papailiopoulos2026} reports a value, by \(0.010\)~bits/use at \(d=0.01\), \(0.022\) at
\(d=0.05\), \(0.026\) at \(d=0.1\), \(0.012\) at \(d=0.3\), \(0.004\) at \(d=0.5\), and at least \(0.002\)
for \(d\ge0.7\).
At \(d\in\{0.02,0.03\}\), the best available upper bound~\cite{PintoRibeiro2026ParallelBA} is
\(0.025\)--\(0.033\)~bits/use above the Markov point estimate and too loose to certify a gain.
At \(N=150\) the gain is certified for \(d\le0.4\) and for \(d\in\{0.6,0.9\}\), with \(0.018\)~bits/use
at \(d=0.1\); at \(N=200\) it is certified for \(d\le0.3\), with \(0.013\)~bits/use at \(d=0.1\) and
\(0.010\) at \(d=0.05\).
Consequently, at the tested \((N,d)\) listed above, strands with known boundaries support strictly
higher rates per symbol than the capacity of the unsegmented channel, and among the tested values
the gain is largest for \(0.05\le d\le0.2\) and remains about \(1\%\) of the rate at \(d=0.01\).
Where no gain is certified, the Markov lower bound falls short of \(U_{\rm P}\) by at most
\(0.006\)~bits/use, less than the width of the certified bracket \([L_{\rm P},U_{\rm P}]\) at the same
\(d\); these cases are inconclusive, since \(\CN/N\ge\Cdel\) always holds, and a sharper upper bound on
\(\Cdel\) or a better input would be needed to resolve them.
The gain decreases over the tested lengths and vanishes as \(N\to\infty\), because \(\CN/N\to\Cdel\)
(Theorem~\ref{thm:sandwich}).

\begin{figure*}[!t]
\centering
\includegraphics[width=\textwidth]{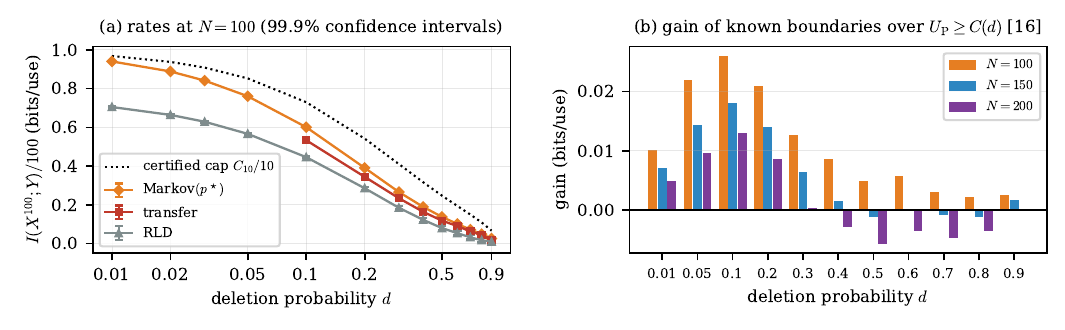}
\caption{Strands of length $100$--$200$.
(a)~$99.9\%$ confidence intervals (Theorem~\ref{thm:conf}) at $N=100$ for the rates of the best
Markov input ($n=40{,}000$), transferred ORD, and RLD ($n=20{,}000$ each), with the cap
$C_{10}(d)/10\ge C_{100}(d)/100$.
(b)~Gain of known strand boundaries: the lower end of the Markov interval, a lower bound on
$\CN/N$, minus the certified upper bound $U_{\rm P}\ge\Cdel$ of~\cite{Papailiopoulos2026}, for
$N\in\{100,150,200\}$ and the values of $d$ for which~\cite{Papailiopoulos2026} reports $U_{\rm P}$;
positive bars are certified gains.}
\label{fig:N100-grid}
\end{figure*}

\begin{table*}[t]
\caption{Strands of length $N\in\{100,150,200\}$ (bits/use): $99.9\%$ confidence intervals from
Theorem~\ref{thm:conf}, the cap $C_{10}(d)/10\ge C_{100}(d)/100$ (BA dual value, rounded up), the certified upper
bound $U\ge\Cdel$ of~\cite{Papailiopoulos2026} ($^\dagger$: of~\cite{PintoRibeiro2026ParallelBA}, where
\cite{Papailiopoulos2026} reports no value), and the certified gain of known boundaries (lower end of
the Markov interval minus $U$, rounded down; ``--'' where it is not positive).
Transfer profiles are available only for $d\ge0.1$.}
\label{tab:N100-conv}
\centering
\scriptsize
\setlength{\tabcolsep}{3pt}
\begin{tabular}{@{}ccccccccccc@{}}
\toprule
 & \multicolumn{4}{c}{$N=100$} & $N=150$ & $N=200$ & & \multicolumn{3}{c}{gain} \\
\cmidrule(lr){2-5}\cmidrule(lr){6-6}\cmidrule(lr){7-7}\cmidrule(l){9-11}
$d$ & Markov$(p^\star)$ & transfer & RLD & cap & Markov$(p^\star)$ & Markov$(p^\star)$ & $U$ & $100$ & $150$ & $200$ \\
\midrule
0.01 & [0.938, 0.943] & -- & [0.702, 0.708] & 0.969 & [0.935, 0.938] & [0.933, 0.936] & 0.929 & 0.0100 & 0.0070 & 0.0049 \\
0.02 & [0.886, 0.891] & -- & [0.661, 0.668] & 0.939 & [0.882, 0.885] & [0.878, 0.881] & 0.915$^\dagger$ & -- & -- & -- \\
0.03 & [0.839, 0.845] & -- & [0.625, 0.634] & 0.909 & [0.834, 0.837] & [0.830, 0.833] & 0.876$^\dagger$ & -- & -- & -- \\
0.05 & [0.758, 0.764] & -- & [0.562, 0.572] & 0.853 & [0.750, 0.754] & [0.746, 0.749] & 0.737 & 0.0218 & 0.0143 & 0.0097 \\
\midrule
0.1 & [0.598, 0.604] & [0.528, 0.540] & [0.440, 0.453] & 0.731 & [0.590, 0.594] & [0.585, 0.588] & 0.573 & 0.0260 & 0.0180 & 0.0130 \\
0.2 & [0.388, 0.394] & [0.337, 0.352] & [0.278, 0.294] & 0.543 & [0.381, 0.384] & [0.375, 0.379] & 0.368 & 0.0209 & 0.0140 & 0.0085 \\
0.3 & [0.263, 0.270] & [0.226, 0.243] & [0.176, 0.195] & 0.411 & [0.257, 0.260] & [0.251, 0.254] & 0.251 & 0.0127 & 0.0064 & 0.0003 \\
0.4 & [0.186, 0.193] & [0.154, 0.174] & [0.113, 0.133] & 0.318 & [0.179, 0.183] & [0.175, 0.178] & 0.179 & 0.0086 & 0.0015 & -- \\
0.5 & [0.134, 0.142] & [0.108, 0.129] & [0.070, 0.090] & 0.248 & [0.129, 0.132] & [0.124, 0.128] & 0.131 & 0.0048 & -- & -- \\
0.6 & [0.098, 0.105] & [0.080, 0.102] & [0.043, 0.064] & 0.195 & [0.092, 0.096] & [0.089, 0.092] & 0.093 & 0.0056 & $<10^{-4}$ & -- \\
0.7 & [0.069, 0.076] & [0.053, 0.074] & [0.022, 0.043] & 0.151 & [0.065, 0.069] & [0.061, 0.065] & 0.067 & 0.0029 & -- & -- \\
0.8 & [0.046, 0.053] & [0.035, 0.055] & [0.007, 0.028] & 0.110 & [0.043, 0.046] & [0.040, 0.043] & 0.045 & 0.0021 & -- & -- \\
0.9 & [0.024, 0.029] & [0.008, 0.027] & [0.000, 0.020] & 0.068 & [0.023, 0.026] & [0.022, 0.024] & 0.023 & 0.0025 & 0.0017 & -- \\
\bottomrule
\end{tabular}

\end{table*}

\subsection{Coding across strands versus coding within one strand}
\label{sec:md}
The block capacity describes codes that span many strands.
A complementary question concerns a code confined to one strand: the largest number \(M\) of
length-\(N\) codewords that a single use of the channel distinguishes with frame error probability
at most \(\varepsilon\).
Morozov and Duman~\cite{MorozovDuman2026FiniteLength} bound this number from above by a
meta-converse with a layered reference output distribution (LO-CVB), computed after a genie splits
the strand into \(n\) segments of length \(m\le32\).
The block capacity yields a second converse for the same quantity.

\begin{proposition}[Single-strand converse]
\label{prop:fano}
Let a code of \(M\) codewords of length \(N\) be used once on the BDC with average error probability
at most \(\varepsilon\le1/2\), and let \(U\ge\Cdel\).
Then
\begin{equation}
\label{eq:fano}
\begin{aligned}
\frac{\log_2M}{N}&\le\frac{\CN/N+h(\varepsilon)/N}{1-\varepsilon}\\
&\le\frac{U+\mathrm{pen}_N(d)+h(\varepsilon)/N}{1-\varepsilon},
\end{aligned}
\end{equation}
where \(h\) is the binary entropy function.
\end{proposition}

\begin{proof}
Let the message \(W\) be uniform on the \(M\) codewords, let \(X^N\) be the codeword of \(W\), and let
\(\hat W\) be the decision made from \(Y\), with error probability \(P_e\le\varepsilon\).
By Fano's inequality~\cite{CoverThomas2006}, \(\HH(W\mid Y)\le h(P_e)+P_e\log_2M\le
h(\varepsilon)+\varepsilon\log_2M\), where the second step uses \(P_e\le\varepsilon\le1/2\).
Since \(W\to X^N\to Y\) is a Markov chain,
\(\log_2M=\I(W;Y)+\HH(W\mid Y)\le\I(X^N;Y)+h(\varepsilon)+\varepsilon\log_2M
\le\CN+h(\varepsilon)+\varepsilon\log_2M\), which gives the first inequality.
Applying Lemma~\ref{lem:fd} to the maximizer of~\eqref{eq:CN} gives \(\CN/N\le\Cdel+\mathrm{pen}_N(d)
\le U+\mathrm{pen}_N(d)\), which gives the second.
\end{proof}

Table~\ref{tab:md} compares~\eqref{eq:fano}, evaluated with \(U=U_{\rm P}\), with the LO-CVB
at \(d=0.2\) and \(\varepsilon=0.2\), the case tabulated in~\cite[Table~III]{MorozovDuman2026FiniteLength}
for \(m=23\) and \(N=23n\).
The LO-CVB is tighter for \(N=23\), the two bounds differ by less than \(0.001\)~bits/use at \(N=46\),
and from \(N=92\) on~\eqref{eq:fano} is smaller by \(0.07\) to \(0.23\)~bits/use.
The reason is that the LO-CVB is computed for an \(m\)-bit genie-aided subchannel and does not decrease
with \(n\), whereas the penalty in~\eqref{eq:fano} decays as \(\log N/N\), so that for fixed
\(\varepsilon\) the bound tends to \(U/(1-\varepsilon)\).
The two converses are therefore complementary.
The last column of Table~\ref{tab:md} gives the certified lower bound on \(\CN/N\), which a code across
many strands achieves with vanishing error probability.
At the tested lengths \(N\in\{184,368,736,1472\}\) and \(\varepsilon=10^{-3}\), the
converse~\eqref{eq:fano} exceeds this lower bound by at most \(0.016\)~bits/use.
Hence, at these lengths, no single-strand code with a small frame error probability can
exceed the rate of coding across strands by more than this margin; whether such codes approach that
rate requires single-strand achievability bounds, which are currently available only for
\(N\le17\)~\cite{MorozovDuman2026FiniteLength}.

\begin{table}[t]
\caption{Single-strand coding at $d=0.2$ (bits/use): the LO-CVB
of~\cite[Table~III]{MorozovDuman2026FiniteLength} ($m=23$, $N=23n$), the converse~\eqref{eq:fano}
with $U=U_{\rm P}$~\cite{Papailiopoulos2026}, and the lower end of the $99.9\%$ Markov interval,
a lower bound on $\CN/N$.}
\label{tab:md}
\centering
\scriptsize
\setlength{\tabcolsep}{3pt}
\begin{tabular}{@{}ccccc@{}}
\toprule
 & \multicolumn{3}{c}{upper bounds on $\log_2M/N$} & lower bound \\
\cmidrule(lr){2-4}
$N$ & LO-CVB, $\varepsilon=0.2$ & \eqref{eq:fano}, $\varepsilon=0.2$ & \eqref{eq:fano}, $\varepsilon=10^{-3}$ & on $\CN/N$ \\
\midrule
23 & 0.548 & 0.660 & 0.498 & 0.460 \\
46 & 0.573 & 0.574 & 0.444 & 0.419 \\
92 & 0.594 & 0.524 & 0.412 & 0.393 \\
184 & 0.622 & 0.495 & 0.393 & 0.377 \\
368 & 0.662 & 0.479 & 0.382 & 0.368 \\
736 & 0.663 & 0.470 & 0.376 & 0.363 \\
1472 & 0.702 & 0.465 & 0.372 & 0.359 \\
\bottomrule
\end{tabular}

\end{table}

\section{Relation to the Unsegmented Channel}
\label{sec:papa}

\begin{table*}[!t]
\caption{Lower bounds on $\Cdel$ at long blocklengths.
For a Markov input with flip probability $p$ and blocklength $N$, $n$ fresh samples give the
estimate $\widehat\I/N$, a $99.9\%$ confidence interval for $\I(X^N;Y)/N$, and the bound $L_\delta$
of~\eqref{eq:cap-conf}, which holds with probability at least $1-10^{-3}$.
Comparison: Drinea--Mitzenmacher (D--M)~\cite{DrineaMitzenmacher2007}, certified run-length bounds
of~\cite{KhodaiemehrFengDuman2026RL}, and the certified enclosure $[L_{\rm P},U_{\rm P}]$
of~\cite{Papailiopoulos2026} (``--'': no value reported).
Lower ends are rounded down and upper ends up.}
\label{tab:papa}
\centering
\scriptsize
\setlength{\tabcolsep}{3pt}
\begin{tabular}{@{}cccccccccccc@{}}
\toprule
$d$ & $p$ & $N$ & $n$ & $\widehat\I/N$ & CI for $\I/N$ & $\mathrm{pen}_N$ & $L_\delta$ & D--M & \cite{KhodaiemehrFengDuman2026RL} & $L_{\rm P}$ & $U_{\rm P}$ \\
\midrule
0.01 & 0.5 & 2000 & 400k & 0.92415 & [0.92400, 0.92430] & 0.00210 & \textbf{0.92191} & 0.92211 & 0.92212 & 0.92211 & 0.92854 \\
0.02 & 0.49 & 2000 & 400k & 0.86696 & [0.86679, 0.86713] & 0.00235 & \textbf{0.86445} & 0.86440 & 0.86456 & -- & -- \\
0.03 & 0.48 & 2000 & 400k & 0.81747 & [0.81728, 0.81765] & 0.00249 & \textbf{0.81480} & 0.81443 & 0.81477 & -- & -- \\
0.05 & 0.47 & 2000 & 400k & 0.73261 & [0.73241, 0.73281] & 0.00267 & \textbf{0.72976} & 0.72829 & 0.72939 & 0.72983 & 0.73646 \\
\midrule
0.1 & 0.438 & 2000 & 200k & 0.56971 & [0.56937, 0.57005] & 0.00290 & \textbf{0.56650} & 0.56196 & 0.56486 & 0.56660 & 0.57245 \\
0.2 & 0.3294 & 1000 & 400k & 0.36187 & [0.36160, 0.36213] & 0.00571 & \textbf{0.35591} & 0.34669 & 0.35127 & 0.35674 & 0.36728 \\
0.3 & 0.246 & 2000 & 100k & 0.23605 & [0.23545, 0.23666] & 0.00320 & \textbf{0.23229} & 0.22243 & 0.22616 & 0.23903 & 0.25084 \\
0.4 & 0.2079 & 2000 & 100k & 0.16012 & [0.15950, 0.16075] & 0.00325 & \textbf{0.15629} & 0.14841 & 0.15136 & 0.16350 & 0.17822 \\
0.5 & 0.154 & 2000 & 100k & 0.11080 & [0.11017, 0.11142] & 0.00326 & \textbf{0.10695} & 0.10186 & 0.10414 & 0.11454 & 0.13017 \\
0.6 & 0.101 & 2000 & 60k & 0.07680 & [0.07592, 0.07769] & 0.00325 & \textbf{0.07273} & 0.06956 & 0.07199 & 0.07960 & 0.09294 \\
0.7 & 0.064 & 2000 & 60k & 0.05135 & [0.05054, 0.05216] & 0.00320 & \textbf{0.04739} & 0.04532 & 0.04796 & 0.05127 & 0.06653 \\
0.8 & 0.0392 & 2000 & 60k & 0.03131 & [0.03061, 0.03200] & 0.00310 & \textbf{0.02755} & 0.02727 & 0.02891 & 0.02884 & 0.04436 \\
0.9 & 0.01248 & 2000 & 60k & 0.01455 & [0.01406, 0.01503] & 0.00290 & \textbf{0.01119} & 0.01238 & 0.01322 & 0.01292 & 0.02218 \\
\bottomrule
\end{tabular}

\end{table*}

\begin{figure*}[!t]
\centering
\includegraphics[width=\textwidth]{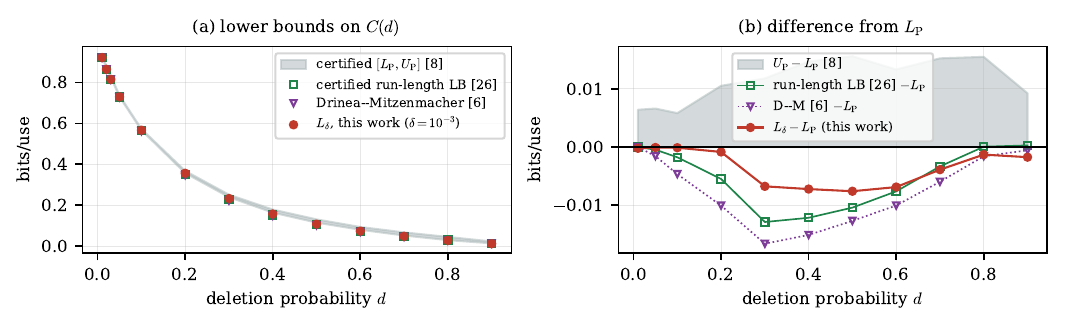}
\caption{(a)~Lower bounds on $\Cdel$: $L_\delta$ from Markov inputs at $N\le2000$ (this work,
$\delta=10^{-3}$), the certified run-length bounds of~\cite{KhodaiemehrFengDuman2026RL}, the
Drinea--Mitzenmacher bound~\cite{DrineaMitzenmacher2007}, and the certified enclosure
$[L_{\rm P},U_{\rm P}]$ of~\cite{Papailiopoulos2026} (shaded).
(b)~The same bounds minus $L_{\rm P}$, where available; shaded: $U_{\rm P}-L_{\rm P}$.}
\label{fig:papa}
\end{figure*}

The block channel is the object of this paper, and Section~\ref{sec:N100} showed that it supports
strictly higher rates than the unsegmented channel at practical strand lengths.
The tools developed for it also yield information in the opposite direction: by
Lemma~\ref{lem:fd}, a block rate minus the length penalty is a lower bound on \(\Cdel\).
Applied to the exact rates at \(N\le16\), this conversion gives weak deterministic bounds, because
the penalty is \(0.14\)--\(0.20\)~bits/use (the strongest, \(0.553\) at \(d=0.1\), is \(0.009\) below the
Drinea--Mitzenmacher bound~\cite{DrineaMitzenmacher2007}); this section therefore applies
Theorem~\ref{thm:conf} at blocklengths up to \(N=2000\), where the penalty is small.

\subsection{Exact marginals for finite-state sources}
Longer blocks require sources whose marginals remain computable, and the recursion of
Proposition~\ref{prop:marginal} extends to any finite-state source.
Let \(X\) be generated by a Markov chain on a finite state set \(\mathcal S\) with emitted bit
\(\beta(s)\), initial law \(\pi_1\), and transition matrix \(T\).
For a fixed output \(y\in\{0,1\}^m\), let
\(F_i(j,s)=\Pr(S_i=s,\ \text{survivors of }X_1^i=y_1^j)\).
Then \(F_1(j,s)=\pi_1(s)[d\,\mathbf 1\{j=0\}+(1-d)\mathbf 1\{j=1,\,y_1=\beta(s)\}]\) and
\begin{equation}
\label{eq:fs-dp}
F_{i+1}(j,s')=d\,G_i(j,s')+(1-d)\,\mathbf 1\{y_j=\beta(s')\}\,G_i(j-1,s'),
\end{equation}
where \(G_i(j,s')=\sum_sF_i(j,s)T(s,s')\), the second term is present only for \(j\ge1\), and
\(F_i(j,s)=0\) for \(j>\min\{i,m\}\).
Hence \(q(y)=\sum_sF_N(m,s)\) costs \(O(N(m+1)E)\) operations, where \(E\) is the number of nonzero
entries of \(T\); this is \(O(N(m+1)|\mathcal S|)\) when each state has two successors.
The estimator~\eqref{eq:plain-mc} remains unbiased with bounded summands, and, for a smooth
fixed-support parameterization \(\theta\) of \(T\), multiplying its summands minus a leave-one-out
baseline by \(\nabla_\theta\ln P_\theta(X_t)\) gives an unbiased estimate of \(\nabla_\theta\I(X^N;Y)\).

Richer finite-state sources, a renewal source with \(24\) states and a run-context source with
\(120\) states (the state space of the largest source in~\cite{Papailiopoulos2026}), learned on the
unbiased gradient at \(N=300\), raise the point estimates of \(\I(X^N;Y)/N\) over the best Markov
input by \(0.002\)--\(0.007\)~bits/use for \(d\in\{0.3,\ldots,0.6\}\).
Our evaluation of \(t(\tau)\) in Theorem~\ref{thm:conf} uses the law of \(-\log_2P(X)\), which we have
not computed for these sources (any valid tail bound would suffice); the bounds below therefore use
Markov inputs.

\subsection{Lower bounds at long blocklengths}
For each \(d\in\{0.01,0.02,0.03,0.05,0.1,\ldots,0.9\}\), we fixed in advance a Markov input, a
blocklength (\(N=2000\), except \(N=1000\) at \(d=0.2\)), and a sample size \(n\); we then drew \(n\) fresh samples and evaluated~\eqref{eq:cap-conf}
with \(\delta=10^{-3}\).
Table~\ref{tab:papa} and Fig.~\ref{fig:papa} compare the resulting bounds \(L_\delta\) with the
Drinea--Mitzenmacher bound~\cite{DrineaMitzenmacher2007}, the certified run-length
bounds~\cite{KhodaiemehrFengDuman2026RL}, and the certified enclosure \([L_{\rm P},U_{\rm P}]\)
of~\cite{Papailiopoulos2026}.
Three observations follow.
First, \(L_\delta\) exceeds the Drinea--Mitzenmacher bound for \(0.02\le d\le0.8\), by up to
\(0.0099\)~bits/use at \(d=0.3\), and the run-length bounds of~\cite{KhodaiemehrFengDuman2026RL} for
\(0.03\le d\le0.6\), by up to \(0.0061\) at \(d=0.3\).
Second, \(L_\delta\) does not exceed \(L_{\rm P}\): the gap is \(1.1\times10^{-4}\) at \(d=0.1\) and
\(8.4\times10^{-4}\) at \(d=0.2\), and it grows to about \(0.007\) for \(0.3\le d\le0.6\), where the
certified source of~\cite{Papailiopoulos2026} is much richer than a first-order Markov chain.
Third, at \(d\ge0.7\) the run-length bounds of~\cite{KhodaiemehrFengDuman2026RL} remain larger, and at
\(d=0.9\) the penalty places \(L_\delta\) below the Drinea--Mitzenmacher bound.
At small deletion probabilities, the regime of DNA storage, all bounds lie within
\(1.6\times10^{-3}\) of each other.
At \(d=0.01\) and \(0.05\), \(L_\delta\) lies \(2.0\times10^{-4}\) and \(7\times10^{-5}\) below \(L_{\rm P}\);
at \(d=0.03\), where~\cite{Papailiopoulos2026} reports no value, \(L_\delta=0.81480\) exceeds the
run-length bound of~\cite{KhodaiemehrFengDuman2026RL} by only \(3\times10^{-5}\); given this margin
and the confidence level, we regard the small-\(d\) bounds as matching the best known ones.

The half-width of the \(99.9\%\) interval for \(\I(X^N;Y)/N\) lies between \(1.5\times10^{-4}\) and
\(8.8\times10^{-4}\).
For \(0.3\le d\le0.6\), the distance to \(L_{\rm P}\) is determined by the input family and by the
length penalty, which decays only as \((\log_2N)/(2N)\); at \(d\in\{0.05,0.1\}\), by contrast, the
point estimate minus the penalty exceeds \(L_{\rm P}\) (by \(2.1\times10^{-4}\) at \(d=0.1\)), and only the
statistical margin places \(L_\delta\) below it.
Each bound holds with probability at least \(1-10^{-3}\), and all thirteen hold simultaneously with
probability at least \(0.987\).
They are not computer-assisted proofs in the sense of~\cite{Papailiopoulos2026}, which relies on
outward-rounded interval arithmetic.

The method of~\cite{Papailiopoulos2026} certifies deterministic lower and upper bounds for all \(d\)
using an A100 GPU and about 216~GB of RAM~\cite[Sec.~13]{Papailiopoulos2026}; the present method
bounds the rate of a given input at any blocklength with prescribed confidence, and the nine bounds
for \(d\ge0.1\) in Table~\ref{tab:papa} required \(3.0\)~CPU-hours with less than 1~GB of memory.
On the other hand, it provides no upper bound on \(\Cdel\), and its guarantee is statistical rather
than deterministic.

\section{Discussion}
\label{sec:accuracy}
\label{sec:limit}

The exact-marginal estimator also quantifies the errors of the earlier version of this work.
All large-\(N\) rates in~\cite{KhodaiemehrFeng2026arXiv} were computed by nested Monte Carlo, which
replaces \(q_w(y)\) in~\eqref{eq:plain-mc} by \(\max\{\bar e_K(y),0.5/K\}\,d^{N-|y|}(1-d)^{|y|}\), where
\(\bar e_K(y)\) is the mean embedding count of \(y\) over an independent pool of \(K\) inputs.
Its bias depends on the input law and has no fixed sign~\cite{Rainforth2018nesting}: compared
with~\eqref{eq:plain-mc}, the earlier estimates for \(16\le N\le512\) were too low by up to
\(0.33\)~bits/use at \(d=0.1\) and too high by up to \(0.28\) at \(d=0.3\), above the certified bound of
Theorem~\ref{thm:sandwich}.
The candidate capacity bounds of that version are withdrawn.

Numerically certified optima are limited to \(N\le16\) for ORD and \(N\le12\) for unrestricted inputs, and the
orbit reduction, although exact, remains exponential (\(K_N\approx2^{N-2}\)).
ORD is uniform within run-count classes, and the transfer profile is heuristic.
The model is binary and contains deletions only, whereas DNA storage uses a quaternary alphabet
and also exhibits substitutions, insertions, and sampling effects~\cite{HeckelMikutisGrass2019,ShomoronyHeckel2021}.
Run counts and the symmetry reduction extend to larger alphabets, and substitutions and insertions
add transitions to the recursion of Proposition~\ref{prop:marginal}; these extensions are left for
future work.
Finally, the length penalty of Lemma~\ref{lem:fd} is paid in full; sharper block conversions would
raise the bounds on \(\Cdel\).

\section{Conclusion}
\label{sec:concl}
This paper studied the block capacity \(\CN\) of the binary deletion channel, the operational limit
for systems that store or transmit many short strands with known boundaries.
For \(d<1\), its maximizer is unique, fully supported, and invariant under complementation and
reversal, and the ORD family attains it within \(2\%\) through \(N=10\), with numerically certified
optima through \(N=16\).
An exact output-marginal recursion turns every simulation into a confidence statement; with it,
learned ORD weights improve on transfer by up to \(0.069\)~bits/use at \(N=128\) and are matched by a
one-parameter Markov input within the statistical precision, whereas InfoNCE estimates saturate at
\(\log_2B\).
For strands of \(100\)--\(200\) bits, known boundaries raise the achievable rate by up to
\(0.026\)~bits/use above the best certified upper bound on \(\Cdel\), the capacity of the unsegmented
channel.
The block capacity also gives an upper bound on the rate of codes that use a single strand, which is
tighter than existing finite-length bounds at the tested lengths \(N\ge92\).
Finally, block rates at \(N\le2000\) give lower bounds on \(\Cdel\) that are within
\(1.1\times10^{-4}\)~bits/use of the best certified lower bound at \(d=0.1\).

\bibliographystyle{IEEEtran}
\bibliography{refs}

\end{document}